\documentclass[12pt]{article}
\usepackage{amssymb,amsmath,amsthm,amscd,latexsym}
\usepackage{mathrsfs}
\usepackage{mathrsfs}
\usepackage{amsfonts}
\usepackage{amsmath}
\usepackage{amssymb}
\usepackage{amscd}
\usepackage{mathrsfs}
\usepackage{amssymb}
\usepackage{amsmath}
\usepackage{amsthm}
\usepackage{color}
\usepackage{latexsym}
\usepackage{indentfirst}
\usepackage{enumitem}
\usepackage{anysize}
\usepackage{bbm}
\usepackage[normalem]{ulem}
\usepackage{soul}
 \usepackage{cite}
\usepackage{cancel}

\renewcommand{\paragraph}{\roman{paragraph}}

\input cyracc.def

\newcommand{\F}{\mathbb{F}}

\theoremstyle{definition}

\newtheorem{thm}{Theorem}%section]
\newtheorem{cor}{Corollary}
\newtheorem{lem}{Lemma}
\newtheorem{prop}{Proposition}

\newtheorem{defn}{Definition}
\newtheorem{rem}{Remark}

\newtheorem{exa}{Example}
\newcommand{\C}{\mathcal{C}}

\newcommand{\wt}{\mathrm{wt}}

\newcommand{\Aut}{\textnormal{Aut}}

\begin{document}
%\begin{CJK*}{GBK}{song}\CJKtilde
\title{\bf A tight upper bound on the number of nonzero weights of a quasi-cyclic code
\thanks{This research is supported by Natural Science Foundation of China under Grant 12071001. The work of Xiaoxiao Li is supported by the China Scholarship Council under Grant 202306500025.}}

\author{
\small{Xiaoxiao Li$^{1}$, Minjia Shi$^{1}$, San Ling$^{2}$}\\ %Patrick Sol\'e$^4$
\and \small{${}^1$School of Mathematical Sciences, Anhui University, Hefei, 230601, China}\\
 \small{${}^2$School of Physical and Mathematical Sciences, Nanyang Technological University, Singapore}\\
%\small{${}^2$Key Laboratory of Intelligent Computing \& Signal Processing,}\\
%\small{Ministry of Education, Anhui University No. 3 Feixi Road,}\\
%\small{Hefei Anhui Province 230039, P. R. China;}\\
}
\date{}
\maketitle
\begin{abstract}
Let $\mathcal{C}$ be a quasi-cyclic code of index $l(l\geq2)$. Let $G$ be the subgroup of the automorphism group of $\mathcal{C}$ generated by $\rho^l$ and the scalar multiplications of $\mathcal{C}$, where $\rho$ denotes the standard cyclic shift. In this paper, we find an explicit formula of orbits of $G$ on $\mathcal{C}\setminus \{\mathbf{0}\}$. Consequently, an explicit upper bound on the number of nonzero weights of $\mathcal{C}$ is immediately derived and a necessary and sufficient condition for codes meeting the bound is exhibited. If $\mathcal{C}$ is a one-generator quasi-cyclic code, a tighter upper bound on the number of nonzero weights of $\mathcal{C}$ is obtained by considering a larger automorphism subgroup which is generated by the multiplier, $\rho^l$ and the scalar multiplications of $\mathcal{C}$. In particular, we list some examples to show the bounds are tight. Our main result improves and generalizes some of the results in \cite{M2}.
\end{abstract}

{\bf Keywords:} Quasi-cyclic code, Hamming weight, upper bound, group action
\section{Introduction}
In 1973, Delsarte studied the for a given code $C$, the relations between the number of distinct distances for $C$, the number of distinct distances for the dual code $C^\bot$, and the minimum distances of $C$ and $C^\bot$, see \cite{PD}. In that paper, some interesting results on the weight distributions of cosets of a code are obtained, which show the importance of the the number of distinct distances in the code. It is easy to see that when one restricts the study to linear codes, then the the number of distinct distances coincides with the number of nonzero weights. The early researches on determining the number of weights of a given linear code can be seen in \cite{TA,TAA,EE,H1,H2,JM}.

For a general linear code, it seems very difficult to obtain an explicit formula for the number of nonzero weights of the code. A more modest goal is to find acceptable bounds on the number of nonzero weights of a linear code. Indeed, there have been several recent works investigating lower and upper bounds on the number of nonzero weights of linear codes. Alderson \cite{TA} determined necessary and sufficient conditions for the existence of full weight spectrum codes, i.e., codes  containing codewords of each weight up to the code length. Shi $et$ $al$. in a series of papers \cite{M1,M2,M3} studied the number of nonzero weights of linear codes. Shi, Li, Neri and Sol$\acute{e}$ \cite{M1} derived upper and lower bounds on the number of nonzero weights of cyclic codes. Chen and Zhang \cite{BZ} obtained the explicit upper bound on the number of nonzero weights of a simple-root cyclic code and exhibit a necessary and sufficient condition for cyclic codes meeting the bound. Moreover, in \cite{BZ}, their results improves and generalizes some of the results in \cite{M1}. Recently, Chen $et$ $al$. \cite{CFL} improved the upper bound in \cite{BZ} with larger subgroups of the automorphism groups of the codes.

Motivated by the work \cite{BZ}, \cite{CFL} and \cite{M2}, the objective of this paper is to establish a tight upper bound on the number of nonzero weights of a quasi-cyclic code of index $l(l\geq2)$ with simple root. In \cite{BZ} and \cite{CFL}, Chen $et$ $al$. pointed out the number of nonzero weights of a linear code is bounded from above by the number of orbits of the automorphism group acting on the code. Let $\mathcal{C}$ be a quasi-cyclic code of length $lm$ and index $l$(co-index $m$). Let $G$ be the subgroup of $\Aut(\mathcal{C})$ (the automorphism group of $\mathcal{C}$) generated by $\rho^l$ and the scalar multiplications of $\mathcal{C}$, where $\rho$ denotes the standard cyclic shift. The problem is therefore converted to finding the number of orbits of $G$ on $\mathcal{C}^*\setminus \{\mathbf{0}\}$. An explicit formula for the number of orbits of $G$ on $\mathcal{C}^*$ is obtained. Consequently, an explicit upper bound on the number of nonzero weights of $\mathcal{C}$ is immediately derived and a necessary and sufficient condition for quasi-cyclic codes meeting the bound is exhibited. If $\mathcal{C}$ is a one-generator quasi-cyclic code, we consider a larger automorphism subgroup which is generated by the multiplier, $\rho^l$ and the scalar multiplications of $\mathcal{C}$ and obtain a tighter upper bound on the number of nonzero weights. We also note that \cite[Section III]{M2} gave some upper bounds on the number of nonzero weights of a special class of strongly quasi-cyclic code, i.e., a quasi-cyclic code of co-index $m$ such that all its nonzero codewords have period $m$. Comparing our results with those in \cite[Section III]{M2}, our results remove the constrain ``strongly" and characterize a necessary and sufficient condition for the codes meeting the bounds.

The material is arranged as follows. Section \ref{sec:2} contains the necessary terminology and definitions on linear codes, quasi-cyclic codes and  group actions. Section \ref{sec:3} presents the main results (see Theorems \ref{thm1}, \ref{thm2} and \ref{thm3}), which give the tight upper bounds on the number of weights that a quasi-cyclic code can have. Section \ref{sec:4} gives the proofs of Theorems \ref{thm1}, \ref{thm2} and \ref{thm3} by counting the number of orbits of $G$ on $\mathcal{C}^*$. Several examples in Section \ref{sec:5} show our bound is tight. Finally, we share our conclusions and some open problems in Section \ref{sec:6}.

\section{Background material}\label{sec:2}
Let $\F_q$ be the finite field with $q$ elements and let $\F^*_q=\F_q\setminus \{0\}$ be the multiplicative group of the finite field $\F_q$. In this section, we review some previously known facts about linear codes, automorphism group of a linear code, and recall some notions and results about quasi-cyclic codes.

\subsection{Linear codes and group actions}
Let $\F_q^n$ be the set of all $n$-tuples whose coordinates belong to $\F_q$. A linear code $\mathcal{C}$ of length $n$ over $\F_q$ is a vector subspace of $\F_q^n$ over $\F_q$. The dimension of the code is its dimension as an $\F_q$-vector space, and is denoted by $k$. A linear code of length $n$ and dimension $k$ over $\F_q$ will be denoted for brevity by $[n, k]$ code. The elements of $\mathcal{C}$ are called codewords.

The Hamming weight of $\mathbf{x}\in \F_q^n$ is the number of indices $i$ where $x_i\neq 0$, and it is denoted by $\wt_H(\mathbf{x})$. The set of weights of a linear code $\mathcal{C}$ (including the $\mathbf{}0$) is denoted by $\wt(\mathcal{C})$, and the number of nonzero weights of $\mathcal{C}$ by $s(\mathcal{C})$, i.e. $\wt(\mathcal{C})=\{\wt_H(\mathbf{c})| \mathbf{c}\in \mathcal{C}\}$ and $s(C)=|\wt(\mathcal{C})\setminus\{0\}|=|\wt(\mathcal{C})|-1$.

\begin{defn}
Let $\mathcal{C}$ be a linear code of length $n$ over $\F_q$. The automorphism group of $\mathcal{C}$, denoted by $\Aut(\mathcal{C})$, consists of all $n\times n$ monomial matrices $A$ over $\F_q$ such that $\mathbf{c}A\in \mathcal{C}$ for all $\mathbf{c}\in \mathcal{C}$.
\end{defn}

Now we recall the result which is the number of nonzero weights of $\mathcal{C}$ is bounded
from the number of $G$-orbits, where $G$ is a subgroup of $\Aut(\mathcal{C})$, see \cite{BZ,JJ}.
\begin{prop}\cite{BZ}\label{p1}
Let $\mathcal{C}$ be a linear code of length $n$ over $\F_q$ with $s(\mathcal{C})$ nonzero weights and let $\Aut(\mathcal{C})$ be the automorphism group of $\mathcal{C}$. Suppose that $G$ is a subgroup of $\Aut(\mathcal{C})$. If the number of orbits of $G$ on $\mathcal{C}^*=\mathcal{C}\setminus\{\mathbf{0}\}$ is equal to $N$, then $s(\mathcal{C})\leq N$.  Moreover, the equality holds if and only if for any two
nonzero codewords $\mathbf{c}_1,\mathbf{c}_2\in \mathcal{C}^*$ with the same weight, there
exists an automorphism $A\in G$ such that $\mathbf{c}_1A=\mathbf{c}_2$.
\end{prop}

In order to determine the number of orbits of $G$ on $\mathcal{C}^*$, we need two important lemmas from \cite{BZ,A}.
\begin{lem}\cite{A}\label{l1}
Let $\mathcal{C}$ be a linear code of length $n$ over $\F_q$ and let $\Aut(\mathcal{C})$ be the automorphism group of $\mathcal{C}$. Suppose that $G$ is a subgroup of $\Aut(\mathcal{C})$. Then, the cardinality of $G\backslash \mathcal{C}^*$ (the set of all the orbits of $G$ on $\mathcal{C}^*$) is equal to
$$|G\backslash \mathcal{C}^*|=\frac{1}{|G|}\sum_{g\in G}|\mbox{Fix}(g)|,$$
where $\mbox{Fix}(g)=\{\mathbf{c}\in \mathcal{C}|g\mathbf{c}=\mathbf{c}\}$.
\end{lem}

\begin{lem}\cite{BZ}\label{l2}
Let $G$ be a finite group acting on a finite set $X$ and let $H$ be a normal subgroup of $G$. It is clear that $H$ naturally acts on $X$. Suppose the set of $H$-orbits are denoted by $H\backslash X=\{Hx|x\in X\}$. Then the factor group $G/H$ acts on $H\backslash X$ and
$$|G\backslash X|=|(G/H)\backslash(H\backslash X)|.$$
\end{lem}

\subsection{Quasi-cyclic codes}
In this subsection, we recall some definitions and results about quasi-cyclic codes. For more detailed information about cyclic codes and quasi-cyclic codes, readers may refer to \cite{P,BHG,C2,SP,SP2,SP3,SP4,RH,KP,FN,WV}.

Let $a_1,a_2,\ldots,a_r$ be integers, where $r\geq 2$ is a positive integer. Let $\gcd(a_1,a_2,\ldots,a_r)$ be the greatest common divisor of $a_1,a_2,\ldots,a_r$. Let $m$ be a positive integer with $\gcd(m,q)=1$. Let $\F_q[x]$ denote the polynomials in the indeterminate $x$ with coefficients in $\F_q$. Let $\langle x^m-1\rangle$ denote the ideal generated by $x^m-1$ in $\F_q[x]$. Then, we have the quotient ring $R_m=\F_q[x]/\langle x^m-1\rangle$.
%Actually, cyclic codes of length $m$ over $\F_q$ are essentially ideals of $R_m$.

We denote by $\rho$ the standard shift operator on $\F^n_q$. A linear code is said to be quasi-cyclic of index $l$ or $l$-quasi-cyclic code if and only if it is invariant under $\rho^l$. Let $\mathcal{C}$ be a quasi-cyclic code over $\F_q$ of length $n=lm$ and index $l$. Let
$$\mathbf{c}=(c_{00},c_{01},\ldots,c_{0,l-1},c_{10},c_{11},\ldots,c_{1,l-1},\ldots,c_{m-1,0},c_{m-1,1},\ldots,c_{m-1,l-1})$$
denote a codeword in $\mathcal{C}$.

Define a map $\phi$: $\F_q^{lm}\rightarrow R^l_m$ by
$$\phi(\mathbf{c})=(\mathbf{c}_0(x),\mathbf{c}_1(x),\ldots,\mathbf{c}_{l-1}(x))\in R_m^l,$$
where $\mathbf{c}_j(x)=\sum_{i=0}^{m-1}c_{ij}x^i\in R_m$. It is known (cf. \cite{KP}, for instance) that $\phi$ induces a one-to-one correspondence between quasi-cyclic codes over $\F_q$ of index $l$ and length $lm$ and linear codes over $R_m$ of length $l$.

It is well known that every minimal ideal of $R_m$ is generated uniquely by a primitive idempotent of $R_m$, see \cite{WV}. There is a one-to-one correspondence between the primitive idempotents of $R_m$ and the $q$-cyclotomic cosets modulo $m$. Let $m'$ be the order of $q$ modulo $m$, i.e., $m'$ is the least positive integer such that $m$ is a divisor of $q^{m'}-1$. Suppose $\zeta$ is a primitive $m$-th root of unity in $\F_{q^{m'}}$ and there are $s+1$ distinct $q$-cyclotomic cosets $\{\Gamma_j\}_{j=0}^s$ modulo $m$ with $\Gamma_0=\{i_0=0\}$ and $\Gamma_t=\{i_t,i_tq,i_tq^2,\ldots,i_tq^{k_t-1}\}$ for $1\leq t\leq s$, where $k_t$ is the cardinality of the $q$-cyclotomic coset $\Gamma_t$ for $0\leq t\leq s$. Then the quotient ring $\F_{q^{m'}}[x]/\langle x^m-1\rangle$ has exactly $m$ primitive idempotents given by
$$e_i=\frac{1}{m}\sum_{j=0}^{m-1}\zeta^{-ij}x^j \ \ \ \ \mbox{for} \ 0\leq i\leq m-1,$$
see \cite{BHG}. Moreover, $R_m=\F_q[x]/\langle x^m-1\rangle$ has exactly $s$ primitive idempotents given by
$$\varepsilon_t=\sum_{j\in \Gamma_t}e_j \ \ \ \mbox{for} \ \ 0\leq t\leq s.$$
According to \cite[Theorem 4.3.8]{WV}, $R_m$ is the vector space direct sum of the minimal ideals $R_m\varepsilon_t$ for $0\leq t\leq s$, in symbols,
$$R_m=R_m\varepsilon_0\oplus R_m\varepsilon_1\oplus\cdots\oplus R_m\varepsilon_s.$$
Using the Discrete Fourier Transform, we have, for each $0\leq t\leq s$,
\begin{equation}\label{Rt}
R_m\varepsilon_t=\bigg\{\sum_{j=0}^{k_t-1}\big(\sum_{u=0}^{k_t-1}c_u\zeta^{li_tq^j}\big)e_{i_tq^j}|c_j\in \F_q\bigg\}.
\end{equation}
Therefore, $R_m^l$ is the direct sum of $(R_m\varepsilon_t)^l$ for $0\leq t\leq s$, in symbols,
$$(R_m)^l=(R_m\varepsilon_0)^l\oplus (R_m\varepsilon_1)^l\oplus\cdots\oplus (R_m\varepsilon_s)^l.$$
It follows that every $R_m$-linear code $\phi(\mathcal{C})$ of length $l$ can be decomposed as the direct sum
\begin{equation}\label{Rl}
 \phi(\mathcal{C})=C_0\oplus C_1\oplus \cdots \oplus C_s,
\end{equation}
where $C_t$ is a linear code over $R_m\varepsilon_t$ of length $l$ for $0\leq t\leq s$ and $\mathcal{C}$ is a quasi-cyclic code over $\F_q$ of length $n=lm$ and index $l$. Actually, for each $1\leq t\leq s$, $C_t$ is a subset of $(R_m\varepsilon_t)^l$. A quasi-cyclic code $\mathcal{C}$ is one-generator if and only if its generator matrix over $R_m$ contains only one row, see \cite{SP3}.
\section{Statement of main results}\label{sec:3}
In this section we give a tight upper bound on $s(\mathcal{C})$ which is the number of nonzero weights of a quasi-cyclic code $\mathcal{C}$. For a general quasi-cyclic code $\mathcal{C}$, we consider two obvious automorphisms: one is the cyclic shift $\rho^l$ whose $\rho$ is the standard shift operator and $l$ is the index of $\mathcal{C}$, and the other is the scalar multiplication. For a one-generator quasi-cyclic code $\mathcal{C}$, apart from the cyclic shift and the scalar multiplications, we consider that the multiplier $\mu_q$ is also an automorphisms of $\mathcal{C}$. According to Proposition \ref{p1}, if the number of the orbits of the group generated by these three automorphisms on $\mathcal{C}$ can be figured out, then we have a upper bound of $s(\mathcal{C})$, naturally.

The main results of this paper are given below.

\begin{thm}\label{thm1}
Let $\mathcal{C}$ be a quasi-cyclic code of length $lm$ and index $l$ over $\F_q$. Suppose that
$$\mathcal{C}=C_{t_1}\oplus C_{t_2}\oplus\cdots\oplus C_{t_U},$$
where $0\leq t_1<t_2<\cdots<t_U\leq s$, $C_{t_j}$ is a linear code over $R_m\varepsilon_{t_j}$ of length $l$ and also is a $[n=lm,K_{t_j}]$ quasi-cyclic code over $\F_q$ for $1\leq j\leq U$. Suppose that the primitive idempotent $\varepsilon_{t_j}$  corresponds to the $q$-cyclotomic coset $\{i_{t_j},i_{t_j}q,\ldots,i_{t_j}q^{k_{t_j}-1}\}$ for each $1\leq j\leq U$. Then the number of orbits of $\langle \rho^l\rangle$ on $\mathcal{C}^*=\mathcal{C}\setminus \{\mathbf{0}\}$ is equal to
$$\sum_{\{j_1,j_2,\ldots,j_u\}\subseteq \{1,2,\ldots,U\},1\leq j_1<j_2<\cdots<j_u\leq U}\frac{\gcd(m,i_{t_{j_1}},i_{t_{j_2}},\ldots,i_{t_{j_u}})\prod_{v=1}^u(q^{K_{t_{j_v}}}-1)}{m},$$
which is denoted by $N$. In particular,
$$s(\mathcal{C})\leq N,$$
with equality if and only if for any codewords $\mathbf{c}_1,\mathbf{c}_2\in \mathcal{C}^*$ with the same weight, there exists an integer $i$ such that $\rho^{il}(\mathbf{c}_1)=\mathbf{c}_2$.
\end{thm}

Let $U=2$, then the formula in Theorem \ref{thm1} can be concise and clear. As a direct application of Theorem \ref{thm1}, we immediately obtain the following corollary.
\begin{cor}\label{cor1}
Let $\mathcal{C}$ be a quasi-cyclic code of length $lm$ and index $l$ over $\F_q$. Suppose that
$$\mathcal{C}=C_{t_1}\oplus C_{t_2},$$
where $0\leq t_1<t_2\leq s$, $C_{t_j}$ is a linear code over $R_m\varepsilon_{t_j}$ of length $l$ and also is a $[n=lm,K_{t_j}]$ quasi-cyclic code over $\F_q$ for $1\leq j\leq 2$. Suppose that the primitive idempotent $\varepsilon_{t_j}$  corresponds to the $q$-cyclotomic coset $\{i_{t_j},i_{t_j}q,\ldots,i_{t_j}q^{k_{t_j}-1}\}$ for each $1\leq j\leq 2$. Then the number of orbits of $\langle \rho^l\rangle$ on $\mathcal{C}^*=\mathcal{C}\setminus \{\mathbf{0}\}$ is equal to
$$\frac{\gcd(m,i_{t_1},i_{t_2})(q^{K_{t_1}}-1)(q^{K_{t_2}}-1)}{m}+\frac{\gcd(m,i_{t_1})(q^{K_{t_1}}-1)}{m}+\frac{\gcd(m,i_{t_2})(q^{K_{t_2}}-1)}{m}.$$
\end{cor}

Next, we turn to study the action of $\langle \rho^l,M\rangle$ on $\mathcal{C}^*$, where $\rho$ is the standard shift operator and $l$ is the index of $\mathcal{C}$, and $M=\{\sigma_a|a\in \F_q^*\}$ consists of the scalar multiplications on $\mathcal{C}$. It is easy to check that $\sigma_a\rho^l=\rho^l\sigma_a$ for any $a\in \F_q^*$. According to the definitions of $\rho^l$ and $M$, we immediately get the following results.
\begin{lem}\label{lem3}
The subgroup $\langle \rho^l,M\rangle$ of $\Aut(\mathcal{C})$ is the direct product of $\rho^l$ and $M$, that is
$$\langle \rho^l,M\rangle=\langle \rho^l\rangle \times M.$$
In particular, $\langle \rho^l,M\rangle$ is of order $m(q-1)$.
\end{lem}

\begin{thm}\label{thm2}
Let $\mathcal{C}$ be a quasi-cyclic code of length $lm$ and index $l$ over $\F_q$. Suppose that
$$\mathcal{C}=C_{t_1}\oplus C_{t_2}\oplus\cdots\oplus C_{t_U},$$
where $0\leq t_1<t_2<\cdots<t_U\leq s$, $C_{t_j}$ is a linear code over $R_m\varepsilon_{t_j}$ of length $l$ and is also a $[n=lm,K_{t_j}]$ quasi-cyclic code over $\F_q$ for $1\leq j\leq U$. Suppose that the primitive idempotent $\varepsilon_{t_j}$  corresponds to the $q$-cyclotomic coset $\{i_{t_j},i_{t_j}q,\ldots,i_{t_j}q^{k_{t_j}-1}\}$ for each $1\leq j\leq U$. Then the number of orbits of $\langle \rho^l,M\rangle$ on $\mathcal{C}^*=\mathcal{C}\setminus \{\mathbf{0}\}$ is equal to
$$\sum_{\{j_1,j_2,\ldots,j_u\}\subseteq \{1,2,\ldots,U\},1\leq j_1<j_2<\cdots<j_u\leq U}\frac{\gcd(m,i_{t_{j_1}},i_{t_{j_2}},\ldots,i_{t_{j_u}})\prod_{v=1}^u(q^{K_{t_{j_v}}}-1)}{m(q-1)}$$
$$\cdot \gcd\bigg(q-1,\frac{m}{\gcd(m,i_{t_{j_1}})},\ldots,\frac{m}{\gcd(m,i_{t_{j_u}})}\bigg),$$
which is denoted by $N$. In particular,
$$s(\mathcal{C})\leq N,$$
with equality if and only if for any codewords $\mathbf{c}_1,\mathbf{c}_2\in \mathcal{C}^*$ with the same weight, there exists an integer $i$ and an element $a\in \F_q^*$ such that $\rho^{il}(a\mathbf{c}_1)=\mathbf{c}_2$.
\end{thm}

By virtue of Theorem \ref{thm2}, we immediately obtain the following corollary.
\begin{cor}\label{cor2}
Let $\mathcal{C}$ be a quasi-cyclic code of length $lm$ and index $l$ over $\F_q$. Suppose that
$$\mathcal{C}=C_{t_1}\oplus C_{t_2},$$
where $0\leq t_1<t_2\leq s$, $C_{t_j}$ is a linear code over $R_m\varepsilon_{t_j}$ of length $l$ and is also a $[n=lm,K_{t_j}]$ quasi-cyclic code over $\F_q$ for $1\leq j\leq 2$. Suppose that the primitive idempotent $\varepsilon_{t_j}$  corresponds to the $q$-cyclotomic coset $\{i_{t_j},i_{t_j}q,\ldots,i_{t_j}q^{k_{t_j}-1}\}$ for each $1\leq j\leq 2$. Then the number of orbits of $\langle \rho^l,M\rangle$ on $\mathcal{C}^*=\mathcal{C}\setminus \{\mathbf{0}\}$ is equal to
\begin{equation*}
  \begin{split}
     & \frac{\gcd(m,i_{t_1},i_{t_2})(q^{K_{t_1}}-1)(q^{K_{t_2}}-1)}{m(q-1)}\cdot \gcd\bigg(q-1,\frac{m}{\gcd(m,i_{t_1})},\frac{m}{\gcd(m,i_{t_2})}\bigg) \\
      & +\frac{\gcd(m,i_{t_1})(q^{K_{t_1}}-1)}{m(q-1)}\cdot \gcd\bigg(q-1,\frac{m}{\gcd(m,i_{t_1})}\bigg)\\
      & +\frac{\gcd(m,i_{t_2})(q^{K_{t_2}}-1)}{m(q-1)}\cdot \gcd\bigg(q-1,\frac{m}{\gcd(m,i_{t_2})}\bigg).
  \end{split}
\end{equation*}
\end{cor}

The map $\mu_q: x \mapsto x^q$ is a ring isomorphism from $R_m$ onto itself. It can be extended to $R^l_m$ componentwise. Then, we turn to study the action of $\langle \mu_q,\rho^l,M\rangle$ on $\mathcal{C}^*$.

\begin{thm}\label{thm3}
Suppose that $\mathbf{f}_1(x),\mathbf{f}_2(x),\ldots,\mathbf{f}_U(x)\in R_m$. Let $\mathcal{C}$ be a one-generator quasi-cyclic code of length $lm$ and index $l$ over $\F_q$. Suppose that
$$\mathcal{C}=C_{t_1}\oplus C_{t_2}\oplus\cdots\oplus C_{t_U},$$
where $0\leq t_1<t_2<\cdots<t_U\leq s$, $C_{t_j}$ is a linear code over $R_m\varepsilon_{t_j}$ of dimension $1$ and length $l$ with generator matrix $[\mathbf{a}_{j,0}(x),\mathbf{a}_{j,1}(x),\ldots,\mathbf{a}_{j,l-1}(x)]$ over $R_m\varepsilon_{t_j}$, where $\mathbf{a}_{j,v}\in \{\mathbf{0},\mathbf{f}_j(x)\}$ for $0\leq v\leq l-1$, and is also a $[n=lm,k_{t_j}]$ quasi-cyclic code over $\F_q$ for $1\leq j\leq U$. Suppose that the primitive idempotent $\varepsilon_{t_j}$  corresponds to the $q$-cyclotomic coset $\{i_{t_j},i_{t_j}q,\ldots,i_{t_j}q^{k_{t_j}-1}\}$ for each $1\leq j\leq U$. Then the number of orbits of $\langle \mu_q,\rho^l,M\rangle$ on $\mathcal{C}^*=\mathcal{C}\setminus \{\mathbf{0}\}$ is equal to
$$\sum_{\{j_1,j_2,\ldots,j_u\}\subseteq \{1,2,\ldots,U\},1\leq j_1<j_2<\cdots<j_u\leq U}N_{j_1,j_2,\ldots,j_u},$$
where
\begin{equation*}
  \begin{split}
    N_{j_1,j_2,\ldots,j_u}= & \frac{1}{m'm(q-1)}\sum_{r=0}^{m'-1}\gcd\bigg(m,\frac{i_{t_{j_1}}II_{t_{j_1}}}{\gcd(I,I_{t_{j_1}})},\cdots,\frac{i_{t_{j_u}}II_{t_{j_u}}}{\gcd(I,I_{t_{j_u}})},\frac{(i_{t_{j_2}}-i_{t_{j_1}})I_{t_{j_1}}I_{t_{j_2}}}{\gcd(I_{t_{j_1}},I_{t_{j_2}})}, \\
      & \cdots, \frac{(i_{t_{j_u}}-i_{t_{j_1}})I_{t_{j_1}}I_{t_{j_u}}}{\gcd(I_{t_{j_1}},I_{t_{j_u}})},\cdots,\frac{(i_{t_{j_u}}-i_{t_{j_{u-1}}})I_{t_{j_{u-1}}}I_{t_{j_u}}}{\gcd(I_{t_{j_{u-1}}},I_{t_{j_u}})}\bigg)\gcd(I,I_{t_{j_1}},\ldots,{t_{j_u}})\\
      &\cdot\prod_{v=1}^u(q^{\gcd(k_{t_{j_{v}}},r)}-1)
  \end{split}
\end{equation*}
with $I=q-1$ and $I_{t_{j_{v}}}=\frac{q^{k_{t_{j_{v}}}}-1}{q^{\gcd(k_{t_{j_{v}}},r)}-1}$ for $v=1,2,\ldots,u$.

In particular,the number of non-zero weights of $\C$ is less than or equal to the number of orbits of $\langle \mu_q,\rho^l,M\rangle$ on $\C^*$.
\end{thm}

By virtue of Theorem \ref{thm3}, we immediately obtain the following corollary.
\begin{cor}\label{cor3}
Suppose that $\mathbf{f}_1(x),\mathbf{f}_2(x)\in R_m$. Let $\mathcal{C}$ be a one-generator quasi-cyclic code of length $lm$ and index $l$ over $\F_q$. Suppose that
$$\mathcal{C}=C_{t_1}\oplus C_{t_2},$$
where $0\leq t_1<t_2\leq s$, $C_{t_j}$ is a linear code over $R_m\varepsilon_{t_j}$ of dimension $1$ and length $l$ with generator matrix $[\mathbf{a}_{j,0}(x),\mathbf{a}_{j,1}(x),\ldots,\mathbf{a}_{j,l-1}(x)]$ over $R_m\varepsilon_{t_j}$, where $\mathbf{a}_{j,v}\in \{\mathbf{0},\mathbf{f}_j(x)\}$ for $0\leq v\leq l-1$, and is also a $[n=lm,k_{t_j}]$ quasi-cyclic code over $\F_q$ for $1\leq j\leq 2$. Suppose that the primitive idempotent $\varepsilon_{t_j}$  corresponds to the $q$-cyclotomic coset $\{i_{t_j},i_{t_j}q,\ldots,i_{t_j}q^{k_{t_j}-1}\}$ for each $1\leq j\leq 2$. Suppose $k_{t_1}|k_{t_2}$, then the number of orbits of $\langle \mu_q,\rho^l,M\rangle$ on $\mathcal{C}^*=\mathcal{C}\setminus \{\mathbf{0}\}$ is equal to
$$s_{t_1}+s_{t_2}+s_{t_1.t_2},$$
where
$$s_{t_{v}}=\frac{1}{k_{t_{v}}}\sum_{r|k_{t_{v}}}\varphi(\frac{k_{t_{v}}}{r})\gcd(q^r-1,\frac{q^{k_{t_{v}}}-1}{q-1},\frac{i_{t_{v}}(q^{k_{t_{v}}}-1)}{m})\ \ \ \ \mbox{for} \ \ \ \ v=1,2,$$
and
\begin{equation*}
  \begin{split}
   s_{t_1,t_2} = & \frac{1}{m'}\sum_{r=0}^{m'-1}\gcd\bigg(\big(q^{\gcd(k_{t_1},r)}-1)\gcd(q^{\gcd(k_{t_2},r)}-1,\frac{(q^{k_{t_1}}-1)(q^{\gcd(k_{t_2},r)}-1)}{(q-1)(q^{\gcd(k_{t_1},r)}-1)},\\
      & \frac{i_{t_1}(q^{k_{t_1}}-1)(q^{\gcd(k_{t_2},r)}-1)}{m{(q^{\gcd(k_{t_1},r)}-1)}},\frac{i_{t_2}(q^{k_{t_2}}-1)}{m}\big),\frac{(i_{t_2}-i_{t_1})(q^{k_{t_1}}-1)(q^{k_{t_2}}-1)}{m(q-1)}\bigg).
  \end{split}
\end{equation*}
In particular,the number of non-zero weights of $\C$ is less than or equal to the number of orbits of $\langle \mu_q,\rho^l,M\rangle$ on $\C^*$.
\end{cor}

\section{Proofs of main results}\label{sec:4}
This section is divided into four parts. First, we give the statement of some lemmas.
Next, we present the proofs of the main results.
\subsection{Statement of some lemmas}
Recall that $R_m=\F_q[x]/\langle x^m-1\rangle$. We have the following two $\F_q$-linear maps on $R_m^l$, denoted by $\rho^l$ and $\sigma_a$, respectively:
$$\rho^l: R^l_m\rightarrow R^l_m$$
$$\rho^l\bigg(\sum_{i=0}^{m-1}c_{i0}x^i,\sum_{i=0}^{m-1}c_{i1}x^i,\ldots,\sum_{i=0}^{m-1}c_{i,l-1}x^i\bigg)=\bigg(\sum_{i=0}^{m-1}c_{i0}x^{i+1},\sum_{i=0}^{m-1}c_{i1}x^{i+1},\ldots,\sum_{i=0}^{m-1}c_{i,l-1}x^{i+1}\bigg)$$
is a $\F_q$-vector space isomorphism of $R_m^l$, and for any fixed element $a\in \F_q^*$,
$$\sigma_a: R^l_m\rightarrow R^l_m$$
$$\sigma_a\bigg(\sum_{i=0}^{m-1}c_{i0}x^i,\sum_{i=0}^{m-1}c_{i1}x^i,\ldots,\sum_{i=0}^{m-1}c_{i,l-1}x^i\bigg)=\bigg(\sum_{i=0}^{m-1}ac_{i0}x^{i},\sum_{i=0}^{m-1}ac_{i1}x^{i},\ldots,\sum_{i=0}^{m-1}ac_{i,l-1}x^{i}\bigg)$$
is a $\F_q$-vector space isomorphism of $R_m^l$. Both $\rho^l$ and $\sigma_a$ are also linear maps on $\F_q^n$ with $n=lm$, which satisfy that for any element $\mathbf{c}$ of $\F_q^n$ and
$$\mathbf{c}=(c_{00},c_{01},\ldots,c_{0,l-1},c_{10},c_{11},\ldots,c_{1,l-1},\ldots,c_{m-1,0},c_{m-1,1},\ldots,c_{m-1,l-1}),$$
then
$$\rho^l(\mathbf{c})=(c_{10},c_{11},\ldots,c_{1,l-1},c_{20},c_{21},\ldots,c_{2,l-1},\ldots,c_{00},c_{01},\ldots,c_{0,l-1})$$
and
$$\sigma_a(\mathbf{c})=(ac_{00},ac_{01},\ldots,ac_{0,l-1},ac_{10},ac_{11},\ldots,ac_{1,l-1},\ldots,ac_{m-1,0},ac_{m-1,1},\ldots,ac_{m-1,l-1}).$$
The map $\mu_q: x \mapsto x^q$ is a ring isomorphism from $R_m$ onto itself. It can be extended to $R^l_m$ componentwise. Specifically, the multiplier $\mu_q$ defined on $R_m^l$ by
$$\mu_q: R^l_m\rightarrow R^l_m$$
$$\mu_q\bigg(\sum_{i=0}^{m-1}c_{i0}x^i,\ldots,\sum_{i=0}^{m-1}c_{i,l-1}x^i\bigg)=\bigg(\sum_{i=0}^{m-1}c_{i0}x^{qi},\ldots,\sum_{i=0}^{m-1}c_{i,l-1}x^{qi}\bigg)\ \ \ \ \mbox{mod} \ \ \ (x^m-1)$$
is a ring automorphism of $R_m^l$. Since $\gcd(m,q)=1$, the map $\mu_q$ induces a permutation of the coefficients of any polynomial in $R_m$.

For any quasi-cyclic code $\mathcal{C}$ of length $n=lm$ and index $l$, it is readily seen that all $\mu_q$, $\rho^l$ and $\sigma_a$ belong to $\Aut(\mathcal{C})$. We know that $M=\{\sigma_a| a\in \F_q^*\}$ is a subgroup of $\Aut(\mathcal{C})$. Clearly, the subgroup $M$ is cyclic with order $q-1$. Since $\gcd(l,n)=\gcd(l,lm)=l$, $\langle \rho^l\rangle$ is of order $m$. Let $m'$ be the order of $q$ modulo $m$. Therefore, $\langle \mu_q\rangle=\{\mu_q^i| 0\leq i\leq m'-1\}$, i.e., $\langle \mu_q\rangle$ is of order $m'$. The proof of the following Lemma is similar to that in \cite[Lemma 2.2]{CFL}, so we omit it.
\begin{lem}\label{lem5}
The subgroup $\langle \mu_q,\rho^l,M\rangle$ of $\Aut(\mathcal{C})$ is of order $m'm(q-1)$, and each element of $\langle \mu_q,\rho^l,M\rangle$ can be written uniquely as a product $\mu_q^{r_1}\rho^{r_2l}\sigma_a$ for some $0\leq r_1\leq m'-1$, $0\leq r_2\leq m-1$ and $a\in \F_q^*$.
\end{lem}

Firstly, we consider the action of $\rho^l$ on $\mathcal{C}^*$. For each integer $i$ with $0\leq i\leq m-1$, it is easy to check that $|\mbox{Fix}(\rho^{il})|=|\mbox{Fix}(\rho^{\gcd(il,n)})|=|\mbox{Fix}(\rho^{\gcd(i,m)l})|$, where
$$\mbox{Fix}(\rho^{il})=\{\mathbf{c}\in \mathcal{C}^*|\rho^{il}(\mathbf{c})=\mathbf{c}\}.$$
For an integer $r$ with $r|m$, the number of integers $i$ satisfying $0\leq i\leq m-1$ and $\gcd(i,m)=r$ is equal to $\varphi(\frac{m}{r})$, where $\varphi$ is Euler's totient function. By Lemma \ref{l1}, one has
\begin{equation}\label{eq1}
|\langle \rho^l\rangle\setminus \mathcal{C}^*|=\frac{1}{m}\sum_{i=0}^{m-1}|\mbox{Fix}(\rho^{il})|=\frac{1}{m}\sum_{r|m}\varphi(\frac{m}{r})|\mbox{Fix}(\rho^{rl})|.
\end{equation}
\begin{lem}\label{lem1}
Let $\mathcal{C}$ be a $[n=lm,K]$ quasi-cyclic code over $\F_q$ which is a linear code over $R_m\varepsilon_t$. Suppose that the primitive idempotent $\varepsilon_t$ corresponds to the $q$-cyclotomic coset $\{i_t,i_tq,\ldots,i_tq^{k-1}\}$. Then the number of orbits of $\langle \rho^l\rangle$ on $\mathcal{C}^*=\mathcal{C}\setminus \{\mathbf{0}\}$ is equal to
$$\frac{\gcd(m,i_t)(q^K-1)}{m}.$$
In particular,
$$s(\mathcal{C})\leq \frac{\gcd(m,i_t)(q^K-1)}{m},$$
with equality if and only if for any codewords $\mathbf{c}_1,\mathbf{c}_2\in \mathcal{C}^*$ with the same weight, there exists an integer $i$ such that $\rho^{il}(\mathbf{c}_1)=\mathbf{c}_2$.
\end{lem}
\begin{proof}
By Proposition \ref{p1}, it is enough to count the number of orbits of $\langle \rho^l\rangle$ on $\mathcal{C}^*$. By Eq. (\ref{eq1}), we aim to find the value of $|\mbox{Fix}(\rho^{rl})|$, for each divisor $r$ of $m$. To this end, let $r$ be a divisor of $m$ and take a typical nonzero element
$$\mathbf{c}=\big(\mathbf{c}_0(x),\mathbf{c}_1(x),\ldots,\mathbf{c}_{l-1}(x)\big)\in \mathcal{C}^*,$$
where $\mathbf{c}_u(x)\in R_m\varepsilon_t$ for $0\leq u\leq l-1$. By Eq. (\ref{Rt}), for each $0\leq u\leq l-1$,
$$\mathbf{c}_u(x)=\sum_{j=0}^{k-1}(c_{u0}+c_{u1}\zeta^{i_tq^j}+\cdots+c_{u,k-1}\zeta^{(k-1)i_tq^j})e_{i_tq^j}\in R_m\varepsilon_t.$$
Note that $e_{i_tq^j}=\frac{1}{m}\sum_{v=0}^{m-1}\zeta^{-i_tq^jv}x^v$, and thus
\begin{equation*}
\begin{split}
     x^re_{i_tq^j} & =\frac{1}{m}\sum_{v=0}^{m-1}\zeta^{-i_tq^jv}x^{v+r} \\
       & =\zeta^{i_tq^jr}\frac{1}{m}\sum_{v=0}^{m-1}\zeta^{-i_tq^j(v+r)}x^{v+r}\\
       & =\zeta^{i_tq^jr}e_{i_tq^j}.
\end{split}
\end{equation*}
Since $\rho^l(\mathbf{c})=\big(x\mathbf{c}_0(x),x\mathbf{c}_1(x),\ldots,x\mathbf{c}_{l-1}(x)\big)$, then we have
$$\rho^{rl}(\mathbf{c})=\big(x^r\mathbf{c}_0(x),x^r\mathbf{c}_1(x),\ldots,x^r\mathbf{c}_{l-1}(x)\big)$$
and
\begin{equation*}
  \begin{split}
    x^r\mathbf{c}_u(x) & =x^r\bigg(\sum_{j=0}^{k-1}(c_{u0}+c_{u1}\zeta^{i_tq^j}+\cdots+c_{u,k-1}\zeta^{(k-1)i_tq^j})e_{i_tq^j}\bigg) \\
      & =\sum_{j=0}^{k-1}(c_{u0}+c_{u1}\zeta^{i_tq^j}+\cdots+c_{u,k-1}\zeta^{(k-1)i_tq^j})x^re_{i_tq^j}\\
      &= \sum_{j=0}^{k-1}\zeta^{i_tq^jr}(c_{u0}+c_{u1}\zeta^{i_tq^j}+\cdots+c_{u,k-1}\zeta^{(k-1)i_tq^j})e_{i_tq^j},
  \end{split}
\end{equation*}
for $0\leq u\leq l-1$. It follows that $\rho^{rl}(\mathbf{c})=\mathbf{c}$ if and only if $x^r\mathbf{c}_u(x)=\mathbf{c}_u(x)$ for all $0\leq u\leq l-1$ if and only if $\zeta^{i_tq^jr}=1$ for all $0\leq j\leq k-1$. Since $\zeta$ is a primitive $m$-th root of unity and $\gcd(m,q)=1$, $\zeta^{i_tq^jr}=1$ precisely when $m$ is a divisor of $i_tr$ (equivalently, $\frac{m}{r}$ is a divisor of $i_t$). This leads to
$$|\mbox{Fix}(\rho^{rl})|=\left\{
  \begin{array}{ll}
    q^K-1, & \mbox{if} \ \ \frac{m}{r}|i_t\hbox{;} \\
    0, &  \mbox{if} \ \ \frac{m}{r}\nmid i_t\hbox{.}
  \end{array}
\right.$$
By Eq. (\ref{eq1}), the number of orbits of $\langle \rho^l\rangle$ on $\mathcal{C}^*$ is equal to
\begin{equation*}
 \begin{split}
   \frac{1}{m}\sum_{i=0}^{m-1}|\mbox{Fix}(\rho^{il})|&=\frac{1}{m}\sum_{r|m}\varphi(\frac{m}{r})|\mbox{Fix}(\rho^{rl})| \\
     & =\frac{q^K-1}{m}\sum_{r|m,\frac{m}{r}|i_t}\varphi(\frac{m}{r})\\
     & =\frac{\gcd(m,i_t)(q^K-1)}{m}.
 \end{split}
\end{equation*}
The proof is completed.
\end{proof}

Based on Lemma \ref{lem3}, we use the method provided in Lemma \ref{l2} to determine the number of orbits of the group $\langle \rho^l,M\rangle$ acting on the quasi-cyclic code.

\begin{lem}\label{lem4}
Let $\mathcal{C}$ be a $[n=lm,K]$ quasi-cyclic code over $\F_q$ which is a linear code over $R_m\varepsilon_t$. Suppose that the primitive idempotent $\varepsilon_t$ corresponds to the $q$-cyclotomic coset $\{i_t,i_tq,\ldots,i_tq^{k-1}\}$. Then the number of orbits of $\langle \rho^l,M\rangle$ on $\mathcal{C}^*=\mathcal{C}\setminus \{\mathbf{0}\}$ is equal to
$$\frac{\gcd\big(m,(q-1)i_t\big)(q^K-1)}{m(q-1)}.$$
In particular,
$$s(\mathcal{C})\leq \frac{\gcd\big(m,(q-1)i_t\big)(q^K-1)}{m(q-1)},$$
with equality if and only if for any codewords $\mathbf{c}_1,\mathbf{c}_2\in \mathcal{C}^*$ with the same weight, there exists an integer $i$ and an element $a\in \F_q^*$ such that $\rho^{il}(a\mathbf{c}_1)=\mathbf{c}_2$.
\end{lem}
\begin{proof}
It is readily seen that the multiplicative cyclic group $\F_q^*$ is isomorphic to $M$; consequently, $M$ is a cyclic group of order $q-1$. In particular, if $\xi$ is a primitive element of $\F_q$ (namely, the cyclic group $\F_q^*$ is generated by $\xi$), then $\sigma_\xi$ is a generator of $M$. Recall that $\langle \rho^l\rangle\backslash \mathcal{C}^*=\{\langle \rho^l\rangle(\mathbf{c})|\mathbf{c}\in \mathcal{C}^*\}$ denotes the set of orbits of $\langle \rho^l\rangle$ on $\mathcal{C}^*=\mathcal{C}\setminus \{\mathbf{0}\}$, where $\langle \rho^l\rangle(\mathbf{c})=\{\rho^{il}(\mathbf{c})|0\leq i\leq m-1\}$. Then $M$ acts on $\langle \rho^l\rangle\backslash \mathcal{C}^*$ in the following natural way:
$$M\times \langle \rho^l\rangle\backslash \mathcal{C}^*\rightarrow \langle \rho^l\rangle\backslash \mathcal{C}^*$$
$$(\sigma_a,\langle \rho^l\rangle(\mathbf{c}))\mapsto \langle \rho^l\rangle(a\mathbf{c}).$$
By Lemma \ref{l2}, the number of orbits of $M$ on $\mathcal{C}^*$ is equal to the number of orbits of $M$ on $\langle \rho^l\rangle\backslash \mathcal{C}^*$, where the latter is equal to
\begin{equation}\label{eq2}
  |(\langle \rho^l,M\rangle)\backslash \mathcal{C}^*|=\frac{1}{q-1}\sum_{r|(q-1)}\varphi\bigg(\frac{q-1}{r}\bigg)|\mbox{Fix}(\sigma^r_\xi)|
\end{equation}
with $\mbox{Fix}(\sigma^r_\xi)=\{\langle \rho^l\rangle(\mathbf{c})\in \langle \rho^l\rangle\backslash \mathcal{C}^*|\langle \rho^l\rangle(\mathbf{c})=\langle \rho^l\rangle(\xi^r\mathbf{c})\}$. Therefore, our ultimate goal is to calculate the value of $|\mbox{Fix}(\sigma^r_\xi)|$. To this end, as we did in the proof of Lemma \ref{lem1}, let $r$ be a divisor of $q-1$ and take a typical nonzero element
$$\mathbf{c}=\big(\mathbf{c}_0(x),\mathbf{c}_1(x),\ldots,\mathbf{c}_{l-1}(x)\big)\in \mathcal{C}^*,$$
where $\mathbf{c}_u(x)\in R_m\varepsilon_t$ for $0\leq u\leq l-1$. By Eq. (\ref{Rt}), for each $0\leq u\leq l-1$,
$$\mathbf{c}_u(x)=\sum_{j=0}^{k-1}(c_{u0}+c_{u1}\zeta^{i_tq^j}+\cdots+c_{u,k-1}\zeta^{(k-1)i_tq^j})e_{i_tq^j}\in R_m\varepsilon_t.$$
The condition $\langle \rho^l\rangle(\mathbf{c})=\langle \rho^l\rangle(\xi^r\mathbf{c})$ is equivalent to requiring that there exists an integer $z\geq 0$ such that $\rho^{zl}(\mathbf{c})=\xi^r\mathbf{c}$. Simple algebraic calculations show that
$$\rho^{zl}(\mathbf{c})=\big(x^z\mathbf{c}_0(x),x^z\mathbf{c}_1(x),\ldots,x^z\mathbf{c}_{l-1}(x)\big)$$
and
$$\xi^r\mathbf{c}=\big(\xi^r\mathbf{c}_0(x),\xi^r\mathbf{c}_1(x),\ldots,\xi^r\mathbf{c}_{l-1}(x)\big),$$
where for each $0\leq u\leq l-1$,
$$x^z\mathbf{c}_u(x)=\sum_{j=0}^{k-1}\zeta^{i_tq^jz}(c_{u0}+c_{u1}\zeta^{i_tq^j}+\cdots+c_{u,k-1}\zeta^{(k-1)i_tq^j})e_{i_tq^j}$$
and
$$\xi^r\mathbf{c}_u(x)=\sum_{j=0}^{k-1}\xi^r(c_{u0}+c_{u1}\zeta^{i_tq^j}+\cdots+c_{u,k-1}\zeta^{(k-1)i_tq^j})e_{i_tq^j}.$$
Therefore, $\rho^{zl}(\mathbf{c})=\xi^r\mathbf{c}$ if and only if $x^z\mathbf{c}_u(x)=\xi^r\mathbf{c}_u(x)$ for $0\leq u\leq l-1$ if and only if there exists an integer $z\geq 0$ such that $\zeta^{i_tq^jz}=\xi^r$ for $0\leq j\leq k-1$. Since $\xi \in \F_q$, there exists an integer $z\geq 0$ such that $\zeta^{i_tq^jz}=\xi^r$ for $0\leq j\leq k-1$ if and only if there exists an integer $z\geq 0$ such that $\zeta^{i_tz}=\xi^r$.

In the following we transform the equality $\zeta^{i_tz}=\xi^r$ into numerical conditions. Suppose that $\omega$ is a primitive element of $\F_q^{m'}$, where $m'$ is the least positive integer such that $m'$ is a divisor of $q^{m'}-1$. Denote by $\mbox{ord}(\alpha)$ the order of the element $\alpha\in \F_q^{m'}$. Note that $\zeta$ is a primitive $m$-th root of unity, $\xi$ is a primitive $(q-1)$-th root of unity and $r$ is a divisor $q-1$. Setting $\zeta=\omega^{\frac{q^{m'}-1}{m}}$ and $\xi=\omega^{\frac{q^{m'}-1}{q-1}}$, we have
\begin{equation*}
 \begin{split}
   \zeta^{i_tz}=\xi^r & \Leftrightarrow  \omega^{\frac{(q^{m'}-1)i_tz}{m}}=\omega^{\frac{(q^{m'}-1)r)}{q-1}}\\
     & \Leftrightarrow \langle \omega^{\frac{(q^{m'}-1)r)}{q-1}}\rangle \subseteq \langle \omega^{\frac{(q^{m'}-1)i_t)}{q-1}}\rangle \\
          & \Leftrightarrow \mbox{ord}(\omega^{\frac{(q^{m'}-1)r)}{q-1}}) | \mbox{ord}( \omega^{\frac{(q^{m'}-1)i_t)}{q-1}}) \\
   & \Leftrightarrow \gcd\big(q^{m'}-1,\frac{(q^{m'}-1)i_t}{m}\big)\bigg|\frac{(q^{m'}-1)r}{q-1}\\
   & \Leftrightarrow\frac{(q^{m'}-1)\gcd(m,i_t)}{m}\bigg|\frac{(q^{m'}-1)r}{q-1}\\
   &\Leftrightarrow \frac{q-1}{r}\bigg|\frac{m}{\gcd(m,i_t)},
 \end{split}
\end{equation*}
where $\langle \omega^{\frac{(q^{m'}-1)r)}{q-1}}\rangle$ and $\langle \omega^{\frac{(q^{m'}-1)i_t)}{q-1}}\rangle$ denote the cyclic subgroups of $\F_{q^{m'}}^*$ generated by $\omega^{\frac{(q^{m'}-1)r)}{q-1}}$ and $\omega^{\frac{(q^{m'}-1)i_t)}{q-1}}$, respectively. It follows that there exists an integer $z\geq 0$ such that $\zeta^{i_tz}=\xi^r$ if and only if $\frac{q-1}{r}$ is a divisor of $\frac{m}{\gcd(m,i_t)}$. By Lemma \ref{lem1}, $\langle \rho^l\rangle\backslash \mathcal{C}^*$ has size $\frac{\gcd(m,i_t)(q^K-1)}{m}$; then we have
$$|\mbox{Fix}(\sigma^{r}_\xi)|=\left\{
  \begin{array}{ll}
    \frac{\gcd(m,i_t)(q^K-1)}{m}, & \mbox{if} \ \ \frac{q-1}{r}|\frac{m}{\gcd(m,i_t)}\hbox{;} \\
    0, &  \mbox{if} \ \ \frac{q-1}{r}\nmid \frac{m}{\gcd(m,i_t)}\hbox{.}
  \end{array}
\right.$$
Returning to Eq. (\ref{eq2}), the number of orbits of $\langle \rho^l,M\rangle$ on $\mathcal{C}^*$ is equal to
\begin{equation*}
  \begin{split}
     |(\langle \rho^l,M\rangle)\backslash \mathcal{C}^*| & =\frac{1}{q-1}\sum_{r|(q-1)}\varphi\bigg(\frac{q-1}{r}\bigg)|\mbox{Fix}(\sigma^r_\xi)| \\
       & =\frac{1}{q-1}\sum_{r|(q-1)}\varphi(r)|\mbox{Fix}(\sigma^{\frac{q-1}{r}}_\xi)|\\
       &=\frac{1}{q-1}\sum_{r|(q-1),r|\frac{m}{\gcd(m,i_t)}}\varphi(r)\frac{\gcd(m,i_t)(q^K-1)}{m}\\
       &=\frac{\gcd(m,i_t)(q^K-1)}{m(q-1)}\sum_{r|(q-1),r|\frac{m}{\gcd(m,i_t)}}\varphi(r)\\
       &=\frac{\gcd\big(q-1,\frac{m}{\gcd(m,i_t)}\big)\gcd(m,i_t)(q^K-1)}{m(q-1)}\\
       &=\frac{\gcd\big(m,(q-1)i_t\big)(q^K-1)}{m(q-1)}.
   \end{split}
\end{equation*}
The proof is completed.
\end{proof}

Next, we consider the action of $\langle \mu_q,\rho^l,M\rangle$ on $\mathcal{C}^*$, where $\mathcal{C}$ is a one-generator quasi-cyclic code.
\begin{lem}\label{lem6}
Suppose that $\mathbf{f}(x)\in R_m$. Let $\mathcal{C}$ be a one-generator quasi-cyclic code over $\F_q$ which is a $[l,1]$-linear code over $R_m\varepsilon_t$ with generator matrix $[\mathbf{a}_{0}(x),\mathbf{a}_{1}(x),\ldots,\mathbf{a}_{l-1}(x)]$ over $R_m\varepsilon_{t_j}$, where $\mathbf{a}_{v}\in \{\mathbf{0},\mathbf{f}(x)\}$ for $0\leq v\leq l-1$. Suppose that the primitive idempotent $\varepsilon_t$ corresponds to the $q$-cyclotomic coset $\{i_t,i_tq,\ldots,i_tq^{k-1}\}$. Then the number of orbits of $\langle \mu_q,\rho^l,M\rangle$ on $\mathcal{C}^*=\mathcal{C}\setminus \{\mathbf{0}\}$ is equal to
$$\frac{1}{k}\sum_{r|k}\varphi\big(\frac{k}{r}\big)\gcd\bigg(q^r-1,\frac{q^k-1}{q-1},\frac{i_t(q^k-1)}{m}\bigg).$$
In particular,
$$s(\mathcal{C})\leq \frac{1}{k}\sum_{r|k}\varphi\big(\frac{k}{r}\big)\gcd\bigg(q^r-1,\frac{q^k-1}{q-1},\frac{i_t(q^k-1)}{m}\bigg),$$
with equality if and only if for any codewords $\mathbf{c}_1,\mathbf{c}_2\in \mathcal{C}^*$ with the same weight, there exist integers $i$, $j$ and $a\in \F_q^*$ such that $\mu_q^i\rho^{jl}\sigma_a(\mathbf{c}_1)=\mathbf{c}_2$.
\end{lem}
\begin{proof}
By Proposition \ref{p1}, it is enough to count the number of orbits of $\langle \mu_q,\rho^l,M\rangle$ on $\mathcal{C}^*$. It follows from Eq. (\ref{eq1}) and Lemma \ref{lem5} that
\begin{equation}\label{eq4}
|\langle \mu_q,\rho^l,M\rangle\backslash \mathcal{C}^*|=\frac{1}{m'm(q-1)}\sum_{r_1=0}^{m'-1}\sum_{r_2=0}^{m-1}\sum_{a\in \F_q^*}\big|\{\mathbf{c}\in \mathcal{C}^*|\mu_q^{r_1}\rho^{r_2l}\sigma_a(\mathbf{c})=\mathbf{c}\}\big|.
\end{equation}
Take a typical nonzero element
$$\mathbf{c}=\big(\mathbf{c}_0(x),\mathbf{c}_1(x),\ldots,\mathbf{c}_{l-1}(x)\big)\in \mathcal{C}^*,$$
where $\mathbf{c}_u(x)\in R_m\varepsilon_t$ for $0\leq u\leq l-1$. Since $\mathcal{C}$ is a $[l,1]$-linear code over $R_m\varepsilon_t$, each $\mathbf{c}_u(x)\in \{\mathbf{0},\mathbf{F}(x)\}$ where $\mathbf{F}(x) \in R_m\varepsilon_t$. Therefore, $\mu_q^{r_1}\rho^{r_2l}\sigma_a(\mathbf{c})=\mathbf{c}$ if and only if $\mu_q^{r_1}\rho^{r_2}\sigma_a(\mathbf{F}(x))=\mathbf{F}(x)$. By Eq. (\ref{Rt}),
$$\mathbf{F}(x)=\sum_{j=0}^{k-1}(f_{0}+f_{1}\zeta^{i_tq^j}+\cdots+f_{k-1}\zeta^{(k-1)i_tq^j})e_{i_tq^j}\in R_m\varepsilon_t.$$
Note that $e_{i_tq^j}=\frac{1}{m}\sum_{v=0}^{m-1}\zeta^{-i_tq^jv}x^v$ and $\rho^{r_2}\sigma_a(e_{i_tq^j})=a\zeta^{i_tq^{j_{r_2}}}e_{i_tq^j}$ thus
$$\mu_q^{r_1}\rho^{r_2}\sigma_a(e_{i_tq^j})=a\zeta^{i_tq^{j_{r_2}}}e_{i_tq^{-r_1+j}},$$
where the subscript ${i_tq^{-r_1+j}}$ is calculated modulo $m$. Then we have
\begin{equation*}
  \begin{split}
    \mu_q^{r_1}\rho^{r_2}\sigma_a(\mathbf{F}(x)) & = \mu_q^{r_1}\rho^{r_2}\sigma_a\bigg(\sum_{j=0}^{k-1}(f_{0}+f_{1}\zeta^{i_tq^j}+\cdots+f_{k-1}\zeta^{(k-1)i_tq^j})e_{i_tq^j}\bigg)\\
      &= \sum_{j=0}^{k-1}(f_{0}+f_{1}\zeta^{i_tq^j}+\cdots+f_{k-1}\zeta^{(k-1)i_tq^j})\mu_q^{r_1}\rho^{r_2}\sigma_a(e_{i_tq^j})\\
      &=\sum_{j=0}^{k-1}a\zeta^{i_tq^{j_{r_2}}}(f_{0}+f_{1}\zeta^{i_tq^j}+\cdots+f_{k-1}\zeta^{(k-1)i_tq^j})e_{i_tq^{-r_1+j}}\\
      &=\sum_{j=0}^{k-1}a\zeta^{i_tq^{-r_1+j}q^{{r_1}_{r_2}}}(f_{0}+f_{1}\zeta^{i_tq^{-r_1+j}}+\cdots+f_{k-1}\zeta^{(k-1)i_tq^{-r_1+j}})^{q^{r_1}}e_{i_tq^{-r_1+j}}\\
      &=\sum_{j=0}^{k-1}a\zeta^{i_tq^{r_1+j_{r_2}}}(f_{0}+f_{1}\zeta^{i_tq^j}+\cdots+f_{k-1}\zeta^{(k-1)i_tq^j})^{q^{r_1}}e_{i_tq^j}.
  \end{split}
\end{equation*}
Hence $\mu_q^{r_1}\rho^{r_2}\sigma_a(\mathbf{F}(x))=\mathbf{F}(x)$ if and only if
$$a(f_{0}+f_{1}\zeta^{i_tq^j}+\cdots+f_{k-1}\zeta^{(k-1)i_tq^j})^{q^{r_1-1}}=\zeta^{-i_tq^{r_1+j_{r_2}}}\ \ \ \mbox{for}\ \ \ 0\leq j\leq k-1,$$
which is equivalent to
$$a(f_{0}+f_{1}\zeta^{i_t}+\cdots+f_{k-1}\zeta^{(k-1)i_t})^{q^{r_1-1}}=\zeta^{-i_tq^{{r_1}_{r_2}}}.$$
Since the minimal polynomial of $\zeta^{i_t}$ over $\F_q$ is of degree $k$, the set
$$\{f_{0}+f_{1}\zeta^{i_t}+\cdots+f_{k-1}\zeta^{(k-1)i_t}|f_{v}\in \F_q, 0\leq v\leq k-1\}$$
forms a subfield of $\F_{q^{m'}}$ of size $q^k$. Therefore, the number of $\mathbf{c}\in \mathcal{C}^*$ satisfying $\mu_q^{r_1}\rho^{r_2l}\sigma_a(\mathbf{c})=\mathbf{c}$ is equal to the number of $\mathbf{F}(x)\in R_m\varepsilon_t$ satisfying $\mu_q^{r_1}\rho^{r_2}\sigma_a(\mathbf{F}(x))=\mathbf{F}(x)$, which is equal to the number of $\alpha\in \F^*_{q^k}$ such that $a\alpha^{q^{r_1-1}}=\zeta^{-i_tq^{{r_1}_{r_2}}}$. By the proof of \cite[Theorem 3.1]{CFL}, we have the following three facts:
\begin{enumerate}
  \item The number of $\alpha\in \F_{q^k}^*$ such that $a\alpha^{q^{r_1-1}}=\zeta^{-i_tq^{{r_1}_{r_2}}}$ is equal to $0$ or $q^{\gcd(k,r_1)}-1$.
  \item Let $\F_q$ and $\F^*_{q^k}$ be generated by $\xi$ and $\theta$, respectively. For $0\leq r_1\leq m'-1$, denote $S(r_1)=\{0\leq r_2\leq m-1|\zeta^{-i_tq^{{r_1}_{r_2}}}\in \langle\xi\rangle\langle \theta^{q^{r_1}-1}\rangle\}$, and then $|S(r_1)|=\gcd(m,i_t|\langle\xi\rangle\langle \theta^{q^{r_1}-1}\rangle|)$.
  \item Suppose $r_2\in S(r_1)$ and denote $R(r_1,r_2)=\{0\leq r\leq q-2|\zeta^{-i_tq^{{r_1}_{r_2}}}\in\xi^r\langle \theta^{q^{r_1}-1}\rangle\}$. Then, $|R(r_1,r_2)|=|\langle\xi\rangle \cap\langle \theta^{q^{r_1}-1}\rangle|$.
\end{enumerate}
According to these three facts and the similar calculation as in \cite[Theorem 3.1]{CFL}, we have that
\begin{equation*}
  \begin{split}
    |\langle \mu_q,\rho^l,M\rangle\backslash \mathcal{C}^*| & =\frac{1}{m'm(q-1)}\sum_{r_1=0}^{m'-1}\sum_{r_2\in S(r_1)}\sum_{r_3\in R(r_1,r_2)}(q^{\gcd(k,r_1)}-1)\\
      &=\frac{1}{m'm(q-1)}\sum_{r_1=0}^{m'-1}|S(r_1)||R(r_1,r_2)|(q^{\gcd(k,r_1)}-1)\\
      &=\frac{1}{m'm(q-1)}\sum_{r_1=0}^{m'-1}\gcd(m|\langle\xi\rangle \cap\langle \theta^{q^{r_1}-1}\rangle|,i_t|\langle\xi\rangle||\langle \theta^{q^{r_1}-1}\rangle|)(q^{\gcd(k,r_1)}-1)\\
      &=\frac{1}{k}\sum_{r|k}\varphi\big(\frac{k}{r}\big)\gcd\bigg(q^r-1,\frac{q^k-1}{q-1},\frac{i_t(q^k-1)}{m}\bigg).
  \end{split}
\end{equation*}

The proof is completed.
\end{proof}

\subsection{Proof of Theorem \ref{thm1} and Corollary \ref{cor1}}
\begin{proof}
It is easy to check that $\mathcal{C}^*$ is equal to
$$\bigsqcup_{\{j_1,j_2,\ldots,j_u\}\subseteq \{1,2,\ldots,U\},1\leq<j_1<j_2<\cdots<j_u\leq U}C_{t_{j_1}}\setminus \{\mathbf{0}\}\oplus C_{t_{j_2}}\setminus \{\mathbf{0}\}\oplus\cdots\oplus C_{t_{j_u}}\setminus \{\mathbf{0}\},$$
which is a disjoint union. For all $j_1,j_2,\ldots,j_u$, $C_{t_{j_v}}$ is a linear code over $R_m\varepsilon_{t_{j_v}}$ of length $l$ with $1\leq v\leq u$. Let $s_{j_1j_2\cdots j_u}$ be the number of orbits of $\langle \rho^l\rangle$ acting on
$$C_{t_{j_1}}\setminus \{\mathbf{0}\}\oplus C_{t_{j_2}}\setminus \{\mathbf{0}\}\oplus\cdots\oplus C_{t_{j_u}}\setminus \{\mathbf{0}\},$$
which is denoted by $\mathcal{C}_{j_1j_2\cdots j_u}^\sharp$. Thus the group $\langle \rho^l\rangle$ can act on the set $\mathcal{C}_{j_1j_2\cdots j_u}^\sharp$ in the same way as the group action on $\mathcal{C}$. Then, we have
$$\mathcal{C}^*=\bigsqcup_{\{j_1,j_2,\ldots,j_u\}\subseteq \{1,2,\ldots,U\},1\leq<j_1<j_2<\cdots<j_u\leq U}\mathcal{C}_{j_1j_2\cdots j_u}^\sharp$$
and
$$|\langle \rho^l\rangle \backslash \mathcal{C}^*|=\sum_{\{j_1,j_2,\ldots,j_u\}\subseteq \{1,2,\ldots,U\},1\leq<j_1<j_2<\cdots<j_u\leq U}s_{j_1j_2\cdots j_u}.$$
It is enough to compute the number of orbits of the group $\langle \rho^l\rangle$ acting on $\mathcal{C}_{j_1j_2\cdots j_u}^\sharp$.

According to Eq. (\ref{eq1}), we only need to compute the value of $|\mbox{Fix}(\rho^{rl})|$ for each divisor $r$ of $m$. Let $\mathbf{c}=\mathbf{c}_{t_{j_1}}+\mathbf{c}_{t_{j_1}}+\cdots+\mathbf{c}_{t_{j_u}}\in \mathcal{C}_{j_1j_2\cdots j_u}^\sharp$, where $\mathbf{c}_{t_{j_v}}\in C_{t_{j_v}}\setminus \{\mathbf{0}\}\subseteq (R_m\varepsilon_{t_{j_v}})^l$ for $v=1,2,\ldots,u$. Suppose that for each $v=1,2,\ldots,u$,
$$\mathbf{c}_{t_{j_v}}=\big(\mathbf{c}_{{t_{j_v}},0}(x),\mathbf{c}_{{t_{j_v}},1}(x),\ldots,\mathbf{c}_{t_{j_v},l-1}(x)\big),$$
where $\mathbf{c}_{{t_{j_v}},v'}(x)=\sum_{j=0}^{k_{t_{j_v}}-1}\sum_{u'=0}^{k_{t_{j_v}}-1}c_{v',u',t_{j_v}}\zeta^{u'i_{t_{j_v}}q^j}e_{i_{t_{j_v}}q^j}$ for $0\leq v'\leq l-1$. Then we have
\begin{equation*}
  \begin{split}
    \rho^{rl}(\mathbf{c}) & = \rho^{rl}(\mathbf{c}_{t_{j_1}})+\rho^{rl}(\mathbf{c}_{t_{j_1}})+\cdots+\rho^{rl}(\mathbf{c}_{t_{j_u}})\\
      &= \bigg(x^r\sum_{v=1}^u\mathbf{c}_{{t_{j_v}},0}(x),x^r\sum_{v=1}^u\mathbf{c}_{{t_{j_v}},1}(x),\ldots,x^r\sum_{v=1}^u\mathbf{c}_{{t_{j_v}},l-1}(x)\bigg),
  \end{split}
\end{equation*}
where for each $0\leq v'\leq l-1$,
\begin{equation*}
\begin{split}
  x^r\sum_{v=1}^u\mathbf{c}_{{t_{j_v}},v'}(x) =& \sum_{j=0}^{k_{t_{j_1}}-1}\zeta^{i_{t_{j_1}}q^jr}\sum_{u'=0}^{k_{t_{j_1}}-1}c_{v',u',t_{j_1}}\zeta^{u'i_{t_{j_1}}q^j}e_{i_{t_{j_1}}q^j}+\cdots\\
    & +\sum_{j=0}^{k_{t_{j_u}}-1}\zeta^{i_{t_{j_u}}q^jr}\sum_{u'=0}^{k_{t_{j_u}}-1}c_{v',u',t_{j_u}}\zeta^{u'i_{t_{j_u}}q^j}e_{i_{t_{j_u}}q^j}.
\end{split}
\end{equation*}
Then we can conclude that $\rho^{rl}(\mathbf{c})=\mathbf{c}$ if and only if $\rho^{rl}(\mathbf{c}_{t_{j_1}})+\rho^{rl}(\mathbf{c}_{t_{j_1}})+\cdots+\rho^{rl}(\mathbf{c}_{t_{j_u}})=\mathbf{c}_{t_{j_1}}+\mathbf{c}_{t_{j_1}}+\cdots+\mathbf{c}_{t_{j_u}}$ if and only if
$x^r\sum_{v=1}^u\mathbf{c}_{{t_{j_v}},v'}(x)=\sum_{v=1}^u\mathbf{c}_{{t_{j_v}},v'}(x)$ for $1\leq v'\leq l-1$ if and only if $\zeta^{i_{t_{j_v}}q^jr}=1$ for all $v$ and $j$ if and only if $m|(i_{t_{j_v}}q^jr)$ for all $v$ and $j$ if and only if $m|(i_{t_{j_v}}r)$ for $1\leq v\leq u$ if and only if $\frac{m}{r}|i_{t_{j_v}}$ for $1\leq v\leq u$.
It follows that
$$|\mbox{Fix}(\rho^{rl})|=\left\{
  \begin{array}{ll}
    \prod_{v=1}^u(q^{K_{t_{j_v}}}-1), & \mbox{if} \ \ \frac{m}{r}|i_{t_v}\ \ \mbox{for}\ \mbox{all}\ v=1,2,\ldots,u\hbox{;} \\
    0, &  \mbox{otherwise}\hbox{.}
  \end{array}
\right.$$
Using Eq. (\ref{eq1}), the number of orbits of $\langle \rho^l\rangle$ on $\mathcal{C}_{j_1j_2\cdots j_u}^\sharp$ is
\begin{equation*}
  \begin{split}
  s_{j_1j_2\cdots j_u}& =|\langle \rho^l\rangle\backslash \mathcal{C}_{j_1j_2\cdots j_u}^\sharp| \\
    & = \frac{1}{m}\sum_{i=0}^{m-1}|\mbox{Fix}(\rho^{il})|=\frac{1}{m}\sum_{r|m}\varphi(\frac{m}{r})|\mbox{Fix}(\rho^{rl})|\\
      & =\frac{1}{m}\sum_{r|m,\frac{m}{r}|i_{t_v},1\leq v\leq u}\varphi(\frac{m}{r})\prod_{v=1}^u(q^{K_{t_{j_v}}}-1)\\
      &=\frac{\gcd(m,i_{t_{j_1}},i_{t_{j_2}},\ldots,i_{t_{j_u}})\prod_{v=1}^u(q^{K_{t_{j_v}}}-1)}{m}.
  \end{split}
\end{equation*}
Therefore, the number of orbits of $\langle \rho^l\rangle$ on $\mathcal{C}^*=\mathcal{C}\setminus \{\mathbf{0}\}$ is equal to
$$\sum_{\{j_1,j_2,\ldots,j_u\}\subseteq \{1,2,\ldots,U\},1\leq<j_1<j_2<\cdots<j_u\leq U}\frac{\gcd(m,i_{t_{j_1}},i_{t_{j_2}},\ldots,i_{t_{j_u}})\prod_{v=1}^u(q^{K_{t_{j_v}}}-1)}{m}.$$

Let $U=2$, then we have $$|\langle \rho^l\rangle\backslash \mathcal{C}^*|=|\langle \rho^l\rangle\backslash \mathcal{C}_{t_1,t_2}^\sharp|+s_{t_1}+s_{t_2}.$$
By Lemma \ref{lem1}, we immediately get
$$|\langle \rho^l\rangle\backslash \mathcal{C}_{t_1t_2}^\sharp|=\frac{\gcd(m,i_{t_1},i_{t_2})(q^{K_{t_1}}-1)(q^{K_{t_2}}-1)}{m},$$
$$s_{t_1}=\frac{\gcd(m,i_{t_1})(q^{K_{t_1}}-1)}{m},s_{t_2}=\frac{\gcd(m,i_{t_2})(q^{K_{t_2}}-1)}{m},$$
which gives the desired result.
\end{proof}
\subsection{Proof of Theorem \ref{thm2} and Corollary \ref{cor2}}
\begin{proof}
It is easy to check that $\mathcal{C}^*$ is equal to
$$\bigsqcup_{\{j_1,j_2,\ldots,j_u\}\subseteq \{1,2,\ldots,U\},1\leq<j_1<j_2<\cdots<j_u\leq U}C_{t_{j_1}}\setminus \{\mathbf{0}\}\oplus C_{t_{j_2}}\setminus \{\mathbf{0}\}\oplus\cdots\oplus C_{t_{j_u}}\setminus \{\mathbf{0}\},$$
which is a disjoint union. For all $j_1,j_2,\ldots,j_u$, $C_{t_{j_v}}$ is a linear code over $R_m\varepsilon_{t_{j_v}}$ of length $l$ with $1\leq v\leq u$. Let $s_{j_1j_2\cdots j_u}$ be the number of orbits of $\langle \rho^l, M\rangle$ acting on
$$C_{t_{j_1}}\setminus \{\mathbf{0}\}\oplus C_{t_{j_2}}\setminus \{\mathbf{0}\}\oplus\cdots\oplus C_{t_{j_u}}\setminus \{\mathbf{0}\},$$
which is denoted by $\mathcal{C}_{j_1j_2\cdots j_u}^\sharp$. Thus the group $\langle \rho^l, M\rangle$ can act on the set $\mathcal{C}_{j_1j_2\cdots j_u}^\sharp$ in the same way as the group action on $\mathcal{C}$. Then, we have
$$\mathcal{C}^*=\bigsqcup_{\{j_1,j_2,\ldots,j_u\}\subseteq \{1,2,\ldots,U\},1\leq<j_1<j_2<\cdots<j_u\leq U}\mathcal{C}_{j_1j_2\cdots j_u}^\sharp$$
and
$$|(\langle \rho^l, M\rangle) \backslash \mathcal{C}^*|=\sum_{\{j_1,j_2,\ldots,j_u\}\subseteq \{1,2,\ldots,U\},1\leq<j_1<j_2<\cdots<j_u\leq U}s_{j_1j_2\cdots j_u}.$$
It is enough to compute the number of orbits of the group $\langle \rho^l,M\rangle$ acting on $\mathcal{C}_{j_1j_2\cdots j_u}^\sharp$.

According to Eq. (\ref{eq2}), the number of orbits of $\langle \rho^l,M\rangle$ on $\mathcal{C}_{j_1j_2\cdots j_u}^\sharp$ is equal to
$$  |(\langle \rho^l,M\rangle)\backslash \mathcal{C}_{j_1j_2\cdots j_u}^\sharp|=\frac{1}{q-1}\sum_{r|(q-1)}\varphi(\frac{q-1}{r})|\mbox{Fix}(\sigma^r_\xi)|$$
with $\mbox{Fix}(\sigma^r_\xi)=\{\langle \rho^l\rangle(\mathbf{c})\in \langle \rho^l\rangle\backslash \mathcal{C}_{j_1j_2\cdots j_u}^\sharp|\langle \rho^l\rangle(\mathbf{c})=\langle \rho^l\rangle(\xi^r\mathbf{c})\}$. Therefore, it is enough to calculate the value of $|\mbox{Fix}(\sigma^r_\xi)|$. Note that $\langle \rho^l\rangle(\mathbf{c})=\langle \rho^l\rangle(\xi^r\mathbf{c})$ is equivalent to requiring that there exists an integer $z$ such that $\rho^{zl}(\mathbf{c})=\xi^r\mathbf{c}$.

Let $\mathbf{c}=\mathbf{c}_{t_{j_1}}+\mathbf{c}_{t_{j_2}}+\cdots+\mathbf{c}_{t_{j_u}}\in \mathcal{C}_{j_1j_2\cdots j_u}^\sharp$, where $\mathbf{c}_{t_{j_v}}\in C_{t_{j_v}}\setminus \{\mathbf{0}\}\subseteq (R_m\varepsilon_{t_{j_v}})^l$ for $v=1,2,\ldots,u$. Then $\rho^{zl}(\mathbf{c})=\xi^r\mathbf{c}$ if and only if
\begin{equation}\label{eq3}
  \rho^{zl}(\mathbf{c}_{t_{j_v}})=\xi^r\mathbf{c}_{t_{j_v}} \ \ \ \mbox{for} \ \ \ \ v=1,2,\ldots,u.
\end{equation}
From the proof of Lemma \ref{lem4}, we have that the equalities (\ref{eq3}) hold if and only if
$$\frac{q-1}{r}\bigg|\frac{m}{\gcd(m,i_{t_{j_v}})}\ \ \ \mbox{for} \ \ \ \ v=1,2,\ldots,u.$$
It follows from the proof of Theorem \ref{thm1} that if $\frac{q-1}{r}$ is a divisor $\frac{m}{\gcd(m,i_{t_{j_v}})}$ for $v=1,2,\ldots,u$, then
$$|\mbox{Fix}(\sigma^r_\xi)|=\frac{\gcd(m,i_{t_{j_1}},i_{t_{j_2}},\ldots,i_{t_{j_u}})\prod_{v=1}^u(q^{K_{t_{j_v}}}-1)}{m};$$
otherwise, $|\mbox{Fix}(\sigma^r_\xi)|=0$. Therefore,
\begin{equation*}
  \begin{split}
  s_{j_1j_2\cdots j_u}=&|(\langle\rho^l\rangle\times M)\backslash \mathcal{C}_{j_1j_2\cdots j_u}^\sharp|\\
=& \frac{1}{q-1}\sum_{r|(q-1)}\varphi\bigg(\frac{q-1}{r}\bigg)|\mbox{Fix}(\sigma^r_\xi)| \\
     =&\frac{\gcd(m,i_{t_{j_1}},i_{t_{j_2}},\ldots,i_{t_{j_u}})\prod_{v=1}^u(q^{K_{t_{j_v}}}-1)}{m(q-1)}\cdot\\ &\sum_{r|(q-1),\frac{q-1}{r}\big|\frac{m}{\gcd(m,i_{t_{j_v}})},v=1,2,\ldots,u}\varphi(\frac{q-1}{r})\\
     =&\frac{\gcd(m,i_{t_{j_1}},i_{t_{j_2}},\ldots,i_{t_{j_u}})\prod_{v=1}^u(q^{K_{t_{j_v}}}-1)}{m(q-1)}\cdot\\ &\gcd\bigg(q-1,\frac{m}{\gcd(m,i_{t_{j_1}})},\ldots,\frac{m}{\gcd(m,i_{t_{j_u}})}\bigg).
   \end{split}
\end{equation*}
Therefore, the number of orbits of $\langle \rho^l,M\rangle$ on $\mathcal{C}^*=\mathcal{C}\setminus \{\mathbf{0}\}$ is equal to
$$\sum_{\{j_1,j_2,\ldots,j_u\}\subseteq \{1,2,\ldots,U\},1\leq<j_1<j_2<\cdots<j_u\leq U}\frac{\gcd(m,i_{t_{j_1}},i_{t_{j_2}},\ldots,i_{t_{j_u}})\prod_{v=1}^u(q^{K_{t_{j_v}}}-1)}{m(q-1)}$$
$$\cdot \gcd\bigg(q-1,\frac{m}{\gcd(m,i_{t_{j_1}})},\ldots,\frac{m}{\gcd(m,i_{t_{j_u}})}\bigg).$$

Let $U=2$, we have
$$\langle\rho^l, M\rangle\backslash \mathcal{C}^*=s'_{t_1t_2}+s'_{t_1}+s'_{t_2}.$$
By Lemma \ref{lem4}, we see that
\begin{equation*}
  \begin{split}
    s'_{t_1t_2} =& \frac{\gcd(m,i_{t_1},i_{t_2})(q^{K_{t_1}}-1)(q^{K_{t_2}}-1)}{m(q-1)}\cdot \gcd\bigg(q-1,\frac{m}{\gcd(m,i_{t_1})},\frac{m}{\gcd(m,i_{t_2})}\bigg), \\
     s'_{t_1} =& \frac{\gcd(m,i_{t_1})(q^{K_{t_1}}-1)}{m(q-1)}\cdot \gcd\bigg(q-1,\frac{m}{\gcd(m,i_{t_1})}\bigg),\\
      s'_{t_2} =& \frac{\gcd(m,i_{t_2})(q^{K_{t_2}}-1)}{m(q-1)}\cdot \gcd\bigg(q-1,\frac{m}{\gcd(m,i_{t_2})}\bigg),
  \end{split}
\end{equation*}
giving the desired result.
\end{proof}
\subsection{Proof of Theorem \ref{thm3} and Corollary \ref{cor3}}
It is easy to check that $\mathcal{C}^*$ is equal to
$$\bigsqcup_{\{j_1,j_2,\ldots,j_u\}\subseteq \{1,2,\ldots,U\},1\leq<j_1<j_2<\cdots<j_u\leq U}C_{t_{j_1}}\setminus \{\mathbf{0}\}\oplus C_{t_{j_2}}\setminus \{\mathbf{0}\}\oplus\cdots\oplus C_{t_{j_u}}\setminus \{\mathbf{0}\},$$
which is a disjoint union. For all $j_1,j_2,\ldots,j_u$, $C_{t_{j_v}}$ is a linear code over $R_m\varepsilon_{t_{j_v}}$ of dimension $1$ and length $l$ with $1\leq v\leq u$. Let $s_{j_1j_2\cdots j_u}$ be the number of orbits of $\langle \mu_q,\rho^{l}, M\rangle$ acting on
$$C_{t_{j_1}}\setminus \{\mathbf{0}\}\oplus C_{t_{j_2}}\setminus \{\mathbf{0}\}\oplus\cdots\oplus C_{t_{j_u}}\setminus \{\mathbf{0}\},$$
which is denoted by $\mathcal{C}_{j_1j_2\cdots j_u}^\sharp$. Thus the group $\langle \mu_q,\rho^{l}, M\rangle$ can act on the set $\mathcal{C}_{j_1j_2\cdots j_u}^\sharp$ in the same way as the group action on $\mathcal{C}$. Then, we have
$$\mathcal{C}^*=\bigsqcup_{\{j_1,j_2,\ldots,j_u\}\subseteq \{1,2,\ldots,U\},1\leq<j_1<j_2<\cdots<j_u\leq U}\mathcal{C}_{j_1j_2\cdots j_u}^\sharp$$
and
$$|(\langle \mu_q,\rho^{l}, M\rangle) \backslash \mathcal{C}^*|=\sum_{\{j_1,j_2,\ldots,j_u\}\subseteq \{1,2,\ldots,U\},1\leq<j_1<j_2<\cdots<j_u\leq U}s_{j_1j_2\cdots j_u}.$$
It is enough to compute the number of orbits of the group $\langle \mu_q,\rho^{l}, M\rangle$ acting on $\mathcal{C}_{j_1j_2\cdots j_u}^\sharp$. According to Eq. (\ref{eq4}), the number of orbits of $\langle \mu_q,\rho^{l}, M\rangle$ on $\mathcal{C}_{j_1j_2\cdots j_u}^\sharp$ is equal to
$$  |\langle \mu_q,\rho^l,M\rangle\backslash \mathcal{C}_{j_1j_2\cdots j_u}^\sharp|=\frac{1}{m'm(q-1)}\sum_{r_1=0}^{m'-1}\sum_{r_2=0}^{m-1}\sum_{a\in \F_q^*}\big|\{\mathbf{c}\in \mathcal{C}_{j_1j_2\cdots j_u}^\sharp|\mu_q^{r_1}\rho^{r_2l}\sigma_a(\mathbf{c})=\mathbf{c}\}\big|.$$

Let $\mathbf{c}=\mathbf{c}_{t_{j_1}}+\mathbf{c}_{t_{j_2}}+\cdots+\mathbf{c}_{t_{j_u}}\in \mathcal{C}_{j_1j_2\cdots j_u}^\sharp$, where $\mathbf{c}_{t_{j_v}}\in C_{t_{j_v}}\setminus \{\mathbf{0}\}\subseteq (R_m\varepsilon_{t_{j_v}})^l$ for $v=1,2,\ldots,u$. Then $\mu_q^{r_1}\rho^{r_2l}\sigma_a(\mathbf{c})=\mathbf{c}$ if and only if
\begin{equation}\label{eq5}
  \mu_q^{r_1}\rho^{r_2l}\sigma_a(\mathbf{c}_{t_{j_v}})=\mathbf{c}_{t_{j_v}} \ \ \ \mbox{for} \ \ \ \ v=1,2,\ldots,u.
\end{equation}
Since $C_{t_{j_v}}$ is a $[l,1]$-linear code over $R_m\varepsilon_{t_{j_v}}$ with generator matrix $$[\mathbf{a}_{j_v,0}(x),\mathbf{a}_{j_v,1}(x),\ldots,\mathbf{a}_{j_v,l-1}(x)],$$
where $\mathbf{a}_{j_v,v'}\in \{\mathbf{0},\mathbf{f}_{j_v}(x)\}$ for $0\leq v'\leq l-1$, each component of $\mathbf{c}_{t_{j_v}}$ is $\mathbf{0}$ or $\mathbf{F}_v(x)$, where $\mathbf{F}_v(x)=\sum_{v'=0}^{k_{t_{j_v}}-1}f_{v,v'}\zeta^{v'i_{t_{j_v}}}\in R_m\varepsilon_{t_{j_v}}$. Hence, the Eq. (\ref{eq5}) is equivalent to
$$\mu_q^{r_1}\rho^{r_2}\sigma_a(\mathbf{F}_v(x))=\mathbf{F}_v(x)\ \ \ \mbox{for} \ \ \ \ v=1,2,\ldots,u,$$
which is equivalent to
$$a(f_{v,0}+f_{v,1}\zeta^{i_{t_{j_v}}}+\cdots+f_{v,k_{t_{j_v}}-1}\zeta^{(k_{t_{j_v}}-1)i_{t_{j_v}}})^{q^{r_1-1}}=\zeta^{-i_{t_{j_v}}q^{{r_1}_{r_2}}} \ \ \mbox{for}  \ \ v=1,2,\ldots,u,$$
by the proof of Lemma \ref{lem6}. For $1\leq v\leq u$, the minimal polynomial of $\zeta^{i_{t_{j_v}}}$ over $\F_q$ is of degree $k_{t_{j_v}}$, and so the set
$$\{f_{v,0}+f_{v,1}\zeta^{i_{t_{j_v}}}+\cdots+f_{v,k_{t_{j_v}}-1}\zeta^{(k_{t_{j_v}}-1)i_{t_{j_v}}}|f_{v,v'}\in \F_q, 0\leq v'\leq k_{t_{j_{v}}}-1\}$$
forms a subfield $\F_{q^{k_{t_{j_{v}}}}}$ of $\F_{q^{m'}}$. Then the number of $\mathbf{c}\in \mathcal{C}_{j_1j_2\cdots j_u}^\sharp$ satisfying $\mu_q^{r_1}\rho^{r_2l}\sigma_a(\mathbf{c})=\mathbf{c}$ is equal to the number of $u$-tuples $(\alpha_{t_{j_1}},\alpha_{t_{j_2}},\ldots,\alpha_{t_{j_u}})$ with $\alpha_{t_{j_v}}\in \F^*_{q^{k_{t_{j_v}}}}$ such that $a\alpha_{t_{j_v}}^{q^{r_1-1}}=\zeta^{-i_{t_{j_v}}q^{{r_1}_{r_2}}}$ for all $1\leq v\leq u$, which is easily checked to be $0$ or $\prod_{v=1}^u(q^{\gcd(k_{t_{j_v}},r_1)}-1)$. By the proof of \cite[Lemma 3.1]{CFL}, we have the following two facts:
\begin{enumerate}
\item For $1\leq v\leq u$, let $\F^*_{q^{k_{j_v}}}$ be generated by $\theta_{t_{j_v}}$. Let $\F_q$ be generated by $\xi$. For $0\leq r_1\leq m'-1$, denote $S(r_1)$ by
    $$\left\{0\leq r_2\leq m-1|\exists 0\leq r_3\leq q-2\ \ s.t.\ \ \zeta^{-i_{t_{j_v}}q^{{r_1}_{r_2}}}\in \langle \theta^{q^{r_1}-1}_{t_{j_v}}\rangle\ \ \mbox{for}\ \ \mbox{all}\ \ 1\leq v\leq u\right\}.$$
    Then,
    \begin{equation*}
      \begin{split}
         |S(r_1)|\leq & \gcd\bigg(m,\frac{i_{t_{j_1}}II_{t_{j_1}}}{\gcd(I,I_{t_{j_1}})},\cdots,\frac{i_{t_{j_u}}II_{t_{j_u}}}{\gcd(I,I_{t_{j_u}})},\frac{(i_{t_{j_2}}-i_{t_{j_1}})I_{t_{j_1}}I_{t_{j_2}}}{\gcd(I_{t_{j_1}},I_{t_{j_2}})}, \\
      & \cdots, \frac{(i_{t_{j_u}}-i_{t_{j_1}})I_{t_{j_1}}I_{t_{j_u}}}{\gcd(I_{t_{j_1}},I_{t_{j_u}})},\cdots,\frac{(i_{t_{j_u}}-i_{t_{j_{u-1}}})I_{t_{j_{u-1}}}I_{t_{j_u}}}{\gcd(I_{t_{j_{u-1}}},I_{t_{j_u}})}\bigg),
       \end{split}
    \end{equation*}
    where $I=q-1$ and $I_{t_{j_{v}}}=\frac{q^{k_{t_{j_{v}}}}-1}{q^{\gcd(k_{t_{j_{v}}},r)}-1}$ for $v=1,2,\ldots,u$.
\item Suppose $r_2\in S(r_1)$ and denote $R(r_1,r_2)$ by
$$\left\{0\leq r_3\leq q-2|\zeta^{-i_{t_{j_v}}q^{{r_1}_{r_2}}}\in \langle \theta^{q^{r_1}-1}_{t_{j_v}}\rangle\ \ \mbox{for}\ \ \mbox{all}\ \ 1\leq v\leq u\right\}.$$
Then, $|R(r_1,r_2)|=\gcd(I,I_{t_{j_1}},\ldots,{t_{j_u}})$.
\end{enumerate}
According to these two facts and the similar calculation as in \cite[Lemma 3.1]{CFL}, we have that
\begin{equation*}
  \begin{split}
    s_{j_1j_2\cdots j_u}= & |\langle \mu_q,\rho^l,M\rangle\backslash \mathcal{C}_{j_1j_2\cdots j_u}^\sharp| \\
    =&\frac{1}{m'm(q-1)}\sum_{r_1=0}^{m'-1}\sum_{r_2\in S(r_1)}\sum_{r_3\in R(r_1,r_2)}\prod_{v=1}^u(q^{\gcd(k_{t_{j_v}},r_1)}-1)\\
      =&\frac{1}{m'm(q-1)}\sum_{r_1=0}^{m'-1}|S(r_1)||R(r_1,r_2)|\prod_{v=1}^u(q^{\gcd(k_{t_{j_v}},r_1)}-1)\\
      \leq & \frac{1}{m'm(q-1)}\sum_{r=0}^{m'-1}\gcd\bigg(m,\frac{i_{t_{j_1}}II_{t_{j_1}}}{\gcd(I,I_{t_{j_1}})},\cdots,\frac{i_{t_{j_u}}II_{t_{j_u}}}{\gcd(I,I_{t_{j_u}})},\frac{(i_{t_{j_2}}-i_{t_{j_1}})I_{t_{j_1}}I_{t_{j_2}}}{\gcd(I_{t_{j_1}},I_{t_{j_2}})}, \\
      & \cdots, \frac{(i_{t_{j_u}}-i_{t_{j_1}})I_{t_{j_1}}I_{t_{j_u}}}{\gcd(I_{t_{j_1}},I_{t_{j_u}})},\cdots,\frac{(i_{t_{j_u}}-i_{t_{j_{u-1}}})I_{t_{j_{u-1}}}I_{t_{j_u}}}{\gcd(I_{t_{j_{u-1}}},I_{t_{j_u}})}\bigg)\gcd(I,I_{t_{j_1}},\ldots,{t_{j_u}})\\
      &\cdot\prod_{v=1}^u(q^{\gcd(k_{t_{j_{v}}},r)}-1),
  \end{split}
\end{equation*}
Therefore, the number of orbits of $\langle \mu_q,\rho^l,M\rangle$ on $\mathcal{C}^*=\mathcal{C}\setminus \{\mathbf{0}\}$ is obtained.

Let $U=2$, we have
$$|\langle \mu_q,\rho^l,M\rangle\backslash \mathcal{C}^*|=s_{t_1}+s_{t_2}+s_{t_1,t_2}.$$
According to the proofs of Lemma \ref{lem6} and \cite[Theorem 3.3]{CFL}, we have
$$s_{t_{v}}=\frac{1}{k_{t_{l'}}}\sum_{r|k_{t_{v}}}\varphi(\frac{k_{t_{v}}}{r})\gcd(q^r-1,\frac{q^{k_{t_{v}}}-1}{q-1},\frac{i_{t_{v}}(q^{k_{t_{v}}}-1)}{m})\ \ \ \ \mbox{for} \ \ \ \ v=1,2,$$
\begin{equation*}
  \begin{split}
   s_{t_1,t_2} = & \frac{1}{m'}\sum_{r=0}^{m'-1}\gcd\bigg(\big(q^{\gcd(k_{t_1},r)}-1)\gcd(q^{\gcd(k_{t_2},r)}-1,\frac{(q^{k_{t_1}}-1)(q^{\gcd(k_{t_2},r)}-1)}{(q-1)(q^{\gcd(k_{t_1},r)}-1)},\\
      & \frac{i_{t_1}(q^{k_{t_1}}-1)(q^{\gcd(k_{t_2},r)}-1)}{m{(q^{\gcd(k_{t_1},r)}-1)}},\frac{i_{t_2}(q^{k_{t_2}}-1)}{m}\big),\frac{(i_{t_2}-i_{t_1})(q^{k_{t_1}}-1)(q^{k_{t_2}}-1)}{m(q-1)}\bigg),
  \end{split}
\end{equation*}
which gives the desired result.

\section{Remarks and examples}\label{sec:5}
\begin{rem}
The reference \cite[Theorem 5]{M2} says that if $\mathcal{C}$ is a $[n=lm,K]$ strongly quasi-cyclic code of co-index $m$ over $\F_q$, then $s(\mathcal{C})\leq \frac{q^K-1}{m}$. If $\gcd(m,i_t)=1$, then Lemma \ref{lem1} generalizes and improves \cite[Theorem 5]{M2} by removing the constrain ``strongly" and characterizing a necessary and sufficient condition for the codes meeting bounds.
\end{rem}

We include three examples to show that the upper bounds given in Lemma \ref{lem1} and Theorem \ref{thm1} are tight.
\begin{exa}\label{e1}
Take $m=9$, $l=2$ and $q=2$ in Lemma \ref{lem1}. All the distinct $2$-cyclotomic cosets modulo $9$ are given by
$$\Gamma_0=\{0\},\Gamma_{1}=\{1,2,4,5,7,8\},\Gamma_2=\{3,6\}.$$
Consider the linear code $\mathcal{C}$ over $R_m\varepsilon_2$, where the primitive idempotent $\varepsilon_2$ corresponds to $\Gamma_2$. Suppose $\zeta$ is a primitive $m$-th root of unity. Actually, let $h(x)=\prod_{r\in \Gamma_2}(x-\zeta^r)=x^2+x+1$, then $g(x)=(x^m-1)/h(x)$ is a generator polynomial of $R_m\varepsilon_2$. Let $[1,g(x)]$ be the generator matrix of $\mathcal{C}$ over $R_m\varepsilon_2$. Then $K=1\cdot |\Gamma_2|=2$. By Lemma \ref{lem1}, we have
$$s(\mathcal{C})\leq \frac{\gcd(m,i_t)(q^K-1)}{m}=\frac{\gcd(9,3)(2^2-1)}{9}=1.$$
Hence, the number of nonzero weights of $\mathcal{C}$ must be equal to $1$. Moreover, Lemma \ref{lem1} also tells us that all the nonzero codewords of $\mathcal{C}$ are in the same $\langle\rho^l\rangle$-orbit.
\end{exa}

\begin{exa}\label{e2}
Take $m=15$, $l=3$ and $q=2$ in Lemma \ref{lem1}. All the distinct $2$-cyclotomic cosets modulo $15$ are given by
$$\Gamma_0=\{0\},\Gamma_{1}=\{1,2,4,8\},\Gamma_2=\{3,6,9,12\},\Gamma_3=\{7,11,13,14\},\Gamma_4=\{5,10\}.$$
Consider the linear code $\mathcal{C}$ over $R_m\varepsilon_0$, where the primitive idempotent $\varepsilon_0$ corresponds to $\Gamma_0$. Actually, let $h(x)=x-1$, then $g(x)=(x^m-1)/h(x)$ is a generator polynomial of $R_m\varepsilon_0$. Let
$$\left(
    \begin{array}{ccc}
      1 & 0 & g(x) \\
      0 & 1 & 0 \\
    \end{array}
  \right)
$$
be the generator matrix of $\mathcal{C}$ over $R_m\varepsilon_0$. Then $K=2\cdot |\Gamma_0|=2$. By Lemma \ref{lem1}, we have
$$s(\mathcal{C})\leq \frac{\gcd(m,i_t)(q^K-1)}{m}=\frac{\gcd(15,0)(2^2-1)}{15}=3.$$
Using the Magma software programming \cite{WJC}, we see that the weight distribution of the quasi-cyclic code $\mathcal{C}$ is $1+x^{15}+x^{30}+x^{45}$, showing that the exact value of $s(\mathcal{C})=3$.
\end{exa}

\begin{exa}\label{e3}
Take $m=9$, $l=2$ and $q=2$ in Theorem \ref{thm1}. All the distinct $2$-cyclotomic cosets modulo $9$ are as shown in Example \ref{e1}. Consider the quasi-cyclic code $\mathcal{C}=C_0\oplus C_2$, where $C_0$ is a linear code over $R_m\varepsilon_0$ and $C_2$ is a linear code over $R_m\varepsilon_2$, where the primitive idempotent $\varepsilon_0$ and $\varepsilon_2$ corresponds to $\Gamma_0$ and $\Gamma_2$, respectively. Actually, let $h_1(x)=x+1$ and $h_2(x)=x^2+x+1$, then $g_1(x)=(x^m-1)/h_1(x)$ and $g_2(x)=(x^m-1)/h_2(x)$ are the generator polynomial of $R_m\varepsilon_0$ and $R_m\varepsilon_2$, respectively. Let $[1,g_i(x)]$ be the generator matrix of $C_i$ over $R_m\varepsilon_i$ with $i=0,2$. Then $K_1=1\cdot |\Gamma_0|=1$ and $K_2=1\cdot |\Gamma_2|=2$. By Theorem \ref{thm1}, we have
$$s(\mathcal{C})\leq \frac{\gcd(9,0,3)(2-1)(2^2-1)}{9}+\frac{\gcd(9,0)(2-1)}{9}+\frac{\gcd(9,3)(2^2-1)}{9}=3.$$
Using the Magma software programming \cite{WJC}, we see that the weight distribution of the quasi-cyclic code $\mathcal{C}$ is $1+3x^{6}+3x^{12}+x^{18}$, showing that the exact value of $s(\mathcal{C})=3$. Moreover, Lemma \ref{lem1} also tells us that any two nonzero codewords of $\mathcal{C}$ with the same weight are in the same $\langle\rho^l\rangle$-orbit.
\end{exa}

\begin{rem}
The reference \cite[Theorem 3]{M2} says that if $\mathcal{C}$ is a $[n=lm,K]$ strongly quasi-cyclic code of co-index $m$ over $\F_q$, then
$$s(\mathcal{C})\leq \frac{l(q^K-1)}{\mbox{lcm}(q-1,n)}=\frac{\gcd(lm,q-1)(q^K-1)}{m(q-1)}.$$
If $\gcd(m,i_t)=1$, then Lemma \ref{lem4} says that
$$s(\mathcal{C})\leq \frac{\gcd(m,q-1)(q^K-1)}{m(q-1)}\leq\frac{\gcd(lm,q-1)(q^K-1)}{m(q-1)}.$$
Therefore, Lemma \ref{lem4} generalizes and improves \cite[Theorem 3]{M2} by removing the constrain ``strongly" and characterizing a necessary and sufficient condition for the codes meeting bounds.
\end{rem}

Next, we also include three examples to show that the upper bounds given in Lemma \ref{lem4} and Theorem \ref{thm2} are tight.
\begin{exa}
Take $m=91$, $l=2$ and $q=9$ in Lemma \ref{lem4}. $\Gamma_2=\{8,72,11\}$ is the $9$-cyclotomic coset modulo $91$ containing $8$. Consider the linear code $\mathcal{C}$ over $R_m\varepsilon_2$, where the primitive idempotent $\varepsilon_2$ corresponds to $\Gamma_2$. Suppose $g(x)$ is a generator polynomial of $R_m\varepsilon_2$. Let $[1,g(x)]$ be the generator matrix of $\mathcal{C}$ over $R_m\varepsilon_2$. Then $K=1\cdot |\Gamma_2|=3$. By Lemma \ref{lem4}, we have
$$s(\mathcal{C})\leq \frac{\gcd(m,(q-1)i_t)(q^K-1)}{m(q-1)}=\frac{\gcd(31,(9-1)8)(9^3-1)}{91(9-1)}=1.$$
Hence, the number of nonzero weights of $\mathcal{C}$ must be equal to $1$. Moreover, Lemma \ref{lem1} also tells us that all the nonzero codewords of $\mathcal{C}$ are in the same $\langle \rho^l,M\rangle$-orbit.
\end{exa}

\begin{exa}
Take $m=39$, $l=2$ and $q=5$ in Lemma \ref{lem4}. $\Gamma_1=\{1,5,8,25\}$ is the $5$-cyclotomic coset modulo $39$ containing $1$. Consider the linear code $\mathcal{C}$ over $R_m\varepsilon_1$, where the primitive idempotent $\varepsilon_1$ corresponds to $\Gamma_1$. Suppose $g(x)$ is a generator polynomial of $R_m\varepsilon_1$. Let $[1,g(x)]$ be the generator matrix of $\mathcal{C}$ over $R_m\varepsilon_2$. Then $K=1\cdot |\Gamma_1|=4$. By Lemma \ref{lem4}, we have
$$s(\mathcal{C})\leq \frac{\gcd(m,(q-1)i_t)(q^K-1)}{m(q-1)}=\frac{\gcd(39,(5-1)1)(5^4-1)}{39(5-1)}=4.$$
Using the Magma software programming \cite{WJC}, we see that the weight distribution of the quasi-cyclic code $\mathcal{C}$ is $1+156x^{59}+156x^{62}+156x^{63}+156x^{66}$, showing that the exact value of $s(\mathcal{C})=4$. Moreover, Lemma \ref{lem1} also tells us that any two nonzero codewords of $\mathcal{C}$ with the same weight are in the same $\langle \rho^l,M\rangle$-orbit.
\end{exa}

\begin{exa}
Take $m=26$, $l=2$ and $q=3$ in Theorem \ref{thm2}. All the distinct $3$-cyclotomic cosets modulo $26$ are given by
$$\Gamma_0=\{0\},\Gamma_{1}=\{1,3,9\},\Gamma_2=\{2,4,6\},\Gamma_3=\{4,10,12\},\Gamma_4=\{5,15,19\},$$
$$\Gamma_5=\{13\},\Gamma_{6}=\{7,11,21\},\Gamma_7=\{8,20,24\},\Gamma_8=\{14,16,22\},\Gamma_9=\{17,23,25\}.$$
Consider the quasi-cyclic code $\mathcal{C}=C_1\oplus C_5$, where $C_1$ is a linear code over $R_m\varepsilon_1$ and $C_5$ is a linear code over $R_m\varepsilon_5$, where the primitive idempotent $\varepsilon_1$ and $\varepsilon_5$ corresponds to $\Gamma_1$ and $\Gamma_5$, respectively. Let $g_1(x)$ and $g_2(x)$ be the generator polynomial of $R_m\varepsilon_1$ and $R_m\varepsilon_5$, respectively. Let $[1,g_1(x)]$ be the generator matrix of $C_1$ over $R_m\varepsilon_1$, and $[0,g_2(x)]$ be the generator matrix of $C_5$ over $R_m\varepsilon_5$. Then $K_1=1\cdot |\Gamma_1|=3$ and $K_2=1\cdot |\Gamma_5|=1$. By Theorem \ref{thm2}, we have
\begin{equation*}
  \begin{split}
     s(\mathcal{C}) \leq & \frac{\gcd(26,1,13)(3^3-1)(3-1)}{26(3-1)}\cdot \gcd\bigg(3-1,\frac{26}{\gcd(26,1)},\frac{26}{\gcd(26,13)}\bigg) \\
      & +\frac{\gcd(26,1)(3^3-1)}{26(3-1)}\cdot \gcd\bigg(3-1,\frac{26}{\gcd(26,1)}\bigg)\\
      & +\frac{\gcd(26,13)(3-1)}{26(3-1)}\cdot \gcd\bigg(3-1,\frac{26}{\gcd(26,13)}\bigg)\\
      =& 4.
  \end{split}
\end{equation*}
Using the Magma software programming \cite{WJC}, we see that the weight distribution of the quasi-cyclic code $\mathcal{C}$ is $1+2x^{26}+26x^{32}+26x^{36}+26x^{38}$, showing that the exact value of $s(\mathcal{C})=4$. Moreover, Lemma \ref{lem1} also tells us that any two nonzero codewords of $\mathcal{C}$ with the same weight are in the same $\langle \rho^l,M\rangle$-orbit.
\end{exa}

\begin{rem}
Let $\mathcal{C}$ be a one-generator quasi-cyclic code over $\F_q$. In Theorem \ref{thm3}, we consider that $\langle \mu_q,\rho^l,M\rangle$ is a subgroup of $Aut(\mathcal{C})$ which is larger than the automorphism groups $\langle \rho^l\rangle$ and $\langle \rho^l,M\rangle$. Therefore, the upper bound in Theorem \ref{thm3} is tighter than that in Theorems \ref{thm1} and \ref{thm2} if $\mathcal{C}$ is a one-generator quasi-cyclic code.
\end{rem}

Next, we also include three examples to show that the upper bounds given in Lemma \ref{lem6} and Theorem \ref{thm3} are tight, and also are compared with that in Lemma \ref{lem4} and Theorem \ref{thm2}.

\begin{exa}
Take $m=11$, $l=2$ and $q=4$ in Lemma \ref{lem6}. All the distinct $4$-cyclotomic cosets modulo $11$ are given by
$$\Gamma_0=\{0\},\Gamma_{1}=\{1,3,4,5,9\},\Gamma_2=\{2,6,7,8,10\}.$$
Consider the linear code $\mathcal{C}$ over $R_m\varepsilon_1$, where the primitive idempotent $\varepsilon_1$ corresponds to $\Gamma_1$. Suppose $g(x)$ is a generator polynomial of $R_m\varepsilon_1$. Let $[0,g(x)]$ be the generator matrix of $\mathcal{C}$ over $R_m\varepsilon_1$. Then $K=1\cdot |\Gamma_1|=5$. By Lemma \ref{lem4}, we have
$$s(\mathcal{C})\leq \frac{\gcd(m,(q-1)i_t)(q^K-1)}{m(q-1)}=\frac{\gcd(11,(4-1)1)(4^5-1)}{11(4-1)}=31.$$
Using Lemma \ref{lem6}, we have
\begin{equation*}
  \begin{split}
    s(\mathcal{C}) & \leq\frac{1}{5}\sum_{r|5}\varphi(\frac{5}{r})\gcd(4^r-1,\frac{4^5-1}{4-1},\frac{1\cdot(4^5-1)}{11}) \\
      & =\frac{1}{5}(\varphi(5)+31\varphi(1))=\frac{1}{5}(4+31)=7.
  \end{split}
\end{equation*}
Using the Magma software programming \cite{WJC}, we see that the weight distribution of the quasi-cyclic code $\mathcal{C}$ is $1+165x^{6}+165x^{7}+165x^{8}+330x^{9}+165x^{10}+33x^{11}$, showing that the exact value of $s(\mathcal{C})=6$.
\end{exa}

\begin{exa}
Take $m=9$, $l=2$ and $q=2$ in Theorem \ref{thm3}. All the distinct $2$-cyclotomic cosets modulo $9$ are as shown in Example \ref{e1}. Consider the quasi-cyclic code $\mathcal{C}=C_0\oplus C_1$, where $C_0$ is a linear code over $R_m\varepsilon_0$ and $C_1$ is a linear code over $R_m\varepsilon_1$, where the primitive idempotent $\varepsilon_0$ and $\varepsilon_1$ corresponds to $\Gamma_0$ and $\Gamma_1$, respectively. Let $g_1(x)$ and $g_2(x)$ be the generator polynomial of $R_m\varepsilon_0$ and $R_m\varepsilon_1$, respectively. Let $[0,g_1(x)]$ be the generator matrix of $C_0$ over $R_m\varepsilon_0$, and $[g_2(x),0]$ be the generator matrix of $C_1$ over $R_m\varepsilon_1$. Then $K_1=1\cdot |\Gamma_0|=1$ and $K_2=1\cdot |\Gamma_1|=6$. By Theorem \ref{thm2}, we have
\begin{equation*}
  \begin{split}
     s(\mathcal{C}) \leq & \frac{\gcd(9,0,1)(2-1)(2^6-1)}{9(2-1)}\cdot \gcd\bigg(2-1,\frac{9}{\gcd(9,0)},\frac{9}{\gcd(9,1)}\bigg) \\
      & +\frac{\gcd(9,0)(2-1)}{9(2-1)}\cdot \gcd\bigg(2-1,\frac{9}{\gcd(9,0)}\bigg)\\
      & +\frac{\gcd(9,1)(2^6-1)}{9(2-1)}\cdot \gcd\bigg(2-1,\frac{9}{\gcd(9,1)}\bigg)\\
      =& 7+1+7=15.
  \end{split}
\end{equation*}
Using Theorem \ref{thm3} and Corollary \ref{cor3}, we have
\begin{equation*}
  \begin{split}
    s(\mathcal{C})\leq & 1+\frac{1}{6}\sum_{r|6}\varphi(\frac{6}{r})\gcd\bigg(2^r-1,\frac{2^6-1}{2-1},\frac{1\cdot(2^6-1)}{9}\bigg) \\
    &+\frac{1}{6}\sum_{r=0}^5\gcd\bigg(\gcd\big(2^{\gcd(6,r)}-1,2^{\gcd(6,r)}-1,\frac{2^6-1}{9}\big),\frac{2^6-1}{9}\bigg)\\
      =& 1+3+3=7.
  \end{split}
\end{equation*}
Using the Magma software programming \cite{WJC}, we see that the weight distribution of the quasi-cyclic code $\mathcal{C}$ is $1+9x^{2}+27x^{4}+27x^{6}+x^{9}+9x^{11}+27x^{13}+27x^{15}$, showing that the exact value of $s(\mathcal{C})=7$. Moreover, Theorem \ref{thm3} also tells us that any two nonzero codewords of $\mathcal{C}$ with the same weight are in the same $\langle \mu_q,\rho^l,M\rangle$-orbit.
\end{exa}

\begin{exa}
Take $m=15$, $l=2$ and $q=2$ in Theorem \ref{thm3}. All the distinct $2$-cyclotomic cosets modulo $15$ are as shown in Example \ref{e2}. Consider the quasi-cyclic code $\mathcal{C}=C_2\oplus C_4$, where $C_2$ is a linear code over $R_m\varepsilon_2$ and $C_4$ is a linear code over $R_m\varepsilon_4$, where the primitive idempotent $\varepsilon_2$ and $\varepsilon_4$ corresponds to $\Gamma_2$ and $\Gamma_4$, respectively. Let $g_1(x)$ and $g_2(x)$ be the generator polynomial of $R_m\varepsilon_2$ and $R_m\varepsilon_4$, respectively. Let $[0,g_1(x)]$ be the generator matrix of $C_2$ over $R_m\varepsilon_2$, and $[g_2(x),0]$ be the generator matrix of $C_4$ over $R_m\varepsilon_4$. Then $K_2=1\cdot |\Gamma_2|=4$ and $K_4=1\cdot |\Gamma_4|=2$. By Theorem \ref{thm2}, we have
\begin{equation*}
  \begin{split}
     s(\mathcal{C}) \leq & \frac{\gcd(15,5,3)(2^2-1)(2^4-1)}{15(2-1)}\cdot \gcd\bigg(2-1,\frac{15}{\gcd(15,5)},\frac{15}{\gcd(15,3)}\bigg) \\
      & +\frac{\gcd(15,5)(2^2-1)}{15(2-1)}\cdot \gcd\bigg(2-1,\frac{15}{\gcd(15,5)}\bigg)\\
      & +\frac{\gcd(15,3)(2^4-1)}{15(2-1)}\cdot \gcd\bigg(2-1,\frac{15}{\gcd(15,3)}\bigg)\\
      =& 3+1+3=7.
  \end{split}
\end{equation*}
Using Theorem \ref{thm3} and Corollary \ref{cor3}, we have
\begin{equation*}
  \begin{split}
    s(\mathcal{C})\leq & \frac{1}{2}\sum_{r|2}\varphi(\frac{2}{r})\gcd\bigg(2^r-1,\frac{2^2-1}{2-1},\frac{5\cdot(2^2-1)}{15}\bigg)+\\
    &\frac{1}{4}\sum_{r|4}\varphi(\frac{4}{r})\gcd\bigg(2^r-1,\frac{2^4-1}{2-1},\frac{3\cdot(2^4-1)}{15}\bigg)+ \\
    &\frac{1}{4}\sum_{r=0}^3\gcd\bigg((2^{\gcd(2,r)}-1)\gcd\big(2^{\gcd(4,r)}-1,\frac{2^{\gcd(4,r)}-1}{2^{\gcd(2,r)}-1},3\big),6\bigg)\\
      =& 1+2+2=5.
  \end{split}
\end{equation*}
Using the Magma software programming \cite{WJC}, we see that the weight distribution of the quasi-cyclic code $\mathcal{C}$ is $1+10x^{6}+3x^{10}+5x^{12}+30x^{16}+15x^{22}$, showing that the exact value of $s(\mathcal{C})=5$. Moreover, Theorem \ref{thm3} also tells us that any two nonzero codewords of $\mathcal{C}$ with the same weight are in the same $\langle \mu_q,\rho^l,M\rangle$-orbit.
\end{exa}
\section{Conclusion}\label{sec:6}
In this paper, we establish an explicit upper bound on the number of nonzero weights of any quasi-cyclic code with simple-root by counting the number of orbits of $\langle \rho^l,M\rangle$ on the code ($\langle \mu_q,\rho^l,M\rangle$ on one-generator quasi-cyclic code); at the same time, we show that a quasi-cyclic code achieves the bound if and only if any two codewords with the same weight are in the same $\langle \rho^l,M\rangle$-orbit ($\langle \mu_q,\rho^l,M\rangle$-orbit). Many examples (see Section \ref{sec:5}) are included to show that our bound is tight. Our main result and its corollaries generalize and improve some of the results in \cite{M2}.

A possible direction for future work is to find tight upper bounds for the number of nonzero weights of quasi-cyclic codes with repeated-root.

\end{document}